\documentclass[preprint, superscriptaddress,
 amsmath,amssymb, aps, prl,usenames,dvipsnames]{revtex4-1}
 
 \usepackage[dvipsnames]{xcolor}
 \usepackage{xr}
 \usepackage[titletoc,toc,title]{appendix}
\usepackage{hyperref,amsthm}
\usepackage{graphicx,epstopdf,bm,amssymb,amsfonts,dcolumn,bm,xcolor,xr}
\hypersetup{colorlinks=true,linkcolor=blue,citecolor=red}

\newcommand{\R}{\mathbb{R}}

\newcommand{\n}{\mathbf{n}}

\renewcommand{\P}{{\mathbb P}}
\newcommand{\E}{{\mathbb E}}

\newtheorem{theorem}{Theorem}
\newtheorem{proposition}[theorem]{Proposition}
\newtheorem{lemma}[theorem]{Lemma}

\def\eps{\varepsilon}
\def\E{\mathbb{E}}
\def\P{\mathbb{P}}
\def\R{\mathbb{R}}

\def\<{\langle}
\def\>{\rangle}
\def\hot{\textup{h.o.t.}}

\def\n{\textbf{n}}
\def\L{\mathcal{L}}

\def\S{\mathbf{S}}

\def\dist{\textup{d}}

\usepackage{mathtools}
\newcommand{\dd}{\textup{d}}
\def\eps{\varepsilon}
\begin{document}	
\title{Fast decisions reflect biases, slow decisions do not}

\author{Samantha Linn}
\email[]{samantha.linn@utah.edu}
\affiliation{Department of Mathematics, University of Utah, Salt Lake City, Utah, USA}
\author{Sean D. Lawley}
\email[]{lawley@math.utah.edu}
\affiliation{Department of Mathematics, University of Utah, Salt Lake City, Utah, USA}
\author{Bhargav R. Karamched}
\email[]{bkaramched@fsu.edu}
\affiliation{Department of Mathematics, Florida State University, Tallahassee, Florida 32306, USA}
\affiliation{Institute of Molecular Biophysics, Florida State University, Tallahassee, Florida 32306, USA}
\affiliation{Program in Neuroscience, Florida State University, Tallahassee, Florida 32306, USA}
\author{Zachary P. Kilpatrick}
\email[]{zpkilpat@colorado.edu}
\affiliation{Department of Applied Mathematics, University of Colorado Boulder, Boulder, Colorado 80309, USA}
\author{Kre\v{s}imir Josi\'{c}}
\email[]{kresimir.josic@gmail.com}
\affiliation{Department of Mathematics, University of Houston, Houston, Texas 77004, USA}
\affiliation{Department of Biology and Biochemistry, University of Houston, Houston, Texas 77004, USA}

%

\date{\today}


\begin{abstract}
Decisions are often made by heterogeneous groups of individuals, each with distinct initial biases and access to information of different quality. We show that in large groups of independent agents who accumulate evidence the first to decide are those with the strongest initial biases. Their decisions align with their initial bias, regardless of the underlying truth. In contrast, agents who decide last make decisions as if they were initially unbiased, and hence make better choices. We obtain asymptotic expressions in the large population limit that quantify how agents' initial inclinations shape early decisions. Our analysis shows how bias, information quality, and decision order interact in non-trivial ways to determine the reliability of decisions in a group.
\end{abstract}

\pacs{Valid PACS appear here}

\maketitle

Evidence accumulation models are used widely to describe how different organisms
integrate information to make choices~\cite{Bogacz2006}.  
Experimental evidence shows that these models capture the dynamics of the decision making process of humans and other animals, including the tradeoff between speed and accuracy~\cite{Ratcliff1978theory,chittka2003bees,newsome1989neuronal,uchida2003speed,swets1961decision}. Such models can also be used to understand how decisions are made in social groups, both when individuals observe each other's choices~\cite{reina2023asynchrony,karamched2020heterogeneity,mann2018collective,tump2022avoiding} and when they act independently~\cite{stickler2023impact}.

The accumulation of evidence is often modeled using biased Brownian motion with the quality of evidence determining the magnitude of drift and diffusion.  An agent is assumed to commit to a decision when the process crosses a threshold.  Most previous evidence accumulation models describe a single agent.  However, questions remain about how the order of choices in a group is related to their accuracy~\cite{zandbelt2014response}.  In a group of initially unbiased individuals accumulating evidence of different quality, the fastest and most accurate decisions are made by those accessing the highest quality information~\cite{reina2023asynchrony}.  Here we ask how the initial biases of individuals in a group impact the order and accuracy of their choices. When is a decision driven mainly by an agent's initial bias as opposed to accumulated evidence?

We show that in large groups of agents starting with different initial biases, early decisions tend to be made by agents with the most extreme predispositions.  The choices of these agents agree with their initial bias, regardless of the quality of the evidence they have access to.   On the other hand, decisions of late deciders do not depend on their initial bias.  Thus, in large groups early decisions reflect only initial inclinations, regardless of which choice is right. Late decisions reflect only accumulated evidence and are more likely to be correct. These effects hold generically,  but not in the special case of initially unbiased agents~\cite{reina2023asynchrony}.

\paragraph{Model description.} We first assume that each individual in a population of $N$ agents has to decide between two choices (hypotheses), $H^+$ and $H^-$.  They do so by accumulating evidence and computing the conditional probabilities, $P(H^{\pm}| \text{evidence})$, that 
one of the two hypotheses is correct. When observations are independent and each provides weak evidence, the log likelihood ratio, or \emph{belief}, of agent $i$ in the group,  $X_i = \log\left(P(H^+| \text{evidence}_i) /
P(H^-| \text{evidence}_i)\right)$, evolves approximately as a biased Brownian motion~\cite{Bogacz2006,Gold02} (See Fig.~\ref{fig1}A),
\begin{equation} \label{E:accumulation}
\dd X_i
=\mu_i\,\dd t+\sqrt{2D_i}\,\dd W_i,
\end{equation}
where the drift, $\mu_i,$ and diffusion coefficient, 
$D_i,$ capture the strength and noisiness of the evidence, respectively~\footnote{For an ideal observer if only measurement noise is present then drift and diffusion are equal in the continuum limit, $\mu_i=D_i=m_i$~\cite{Bogacz2006,veliz16}. We consider the more general case where internal noise can increase $D_i>\mu_i$.}.   For all agents the correct choice ($H \in H^{\pm}$) is given by the sign of the drift (${\rm sign}[\mu_i] = \pm 1$). Eq.~\eqref{E:accumulation} is widely used and accurately captures the dynamics of decisions in humans and  animals, including variability in response time and the impact of evidence quality and
biases on choice~\cite{Ratcliff1978theory,gold2007neural,mulder2012}.

Agents start with an initial bias, $X_i(0)$, reflecting information or assumptions they have about the prior probability of either hypothesis~\cite{mulder2012}. We denote by $y$ the initial data for a generic agent.
Each agent then accumulates evidence, and its beliefs evolve according to Eq.~(\ref{E:accumulation}). Agent $i$ makes a decision when its belief reaches one of two thresholds,
$- \theta < 0 < \theta,$ at \emph{decision  time}  $\tau_i
:=\inf\{t>0:X_i(t)\notin(-\theta,\theta)\}$. This decision, denoted by $d_i = H^{\pm}$, is determined by the sign of the threshold reached, ${\rm sign}[X_i(\tau_i)]$. If decision criteria differ between agents an appropriate rescaling of $X_i(0)$, $\mu_i$, and $D_i$ allows us to assume that all agents use the same thresholds~\cite{Bogacz2006}. 

\paragraph{Agents with the most extreme initial biases decide first.} We show that in large groups agents whose initial biases are closest to one of the thresholds make the earliest decisions. We first assume observers are identical  
except for their initial biases, so that $\mu_i = \mu$ and $D_i = D$ in Eq.~\eqref{E:accumulation}. 
We denote by $T_i$ the $i^{\rm th}$ decision time so that 
$T_1\le T_2\le \cdots\le T_N$, where $T_i = \tau_{n(i)}$  and $n(i)$ is the index of
the $i^{\rm th}$ agent to decide. Hence, the index of the first decider is $n(1)$.


For simplicity, we assume that each agent starts with one of finitely many initial beliefs, $\{x_0,x_1, \ldots, x_{I-1}\}$, sampled with probability
$q_i=P(X_j(0) = x_i)$ for $i=0,...,I-1$. The distance of the initial belief $x_{i}$ to the closest threshold is
$
L_{i}=\min\{\theta-x_{i},x_{i}+\theta\}.
$
Let $i=0$ be the index of the unique most extreme initial belief held by an agent, so $L_{0}<L_i$ for $i\neq 0$.

\begin{figure}
\includegraphics[width = \columnwidth]{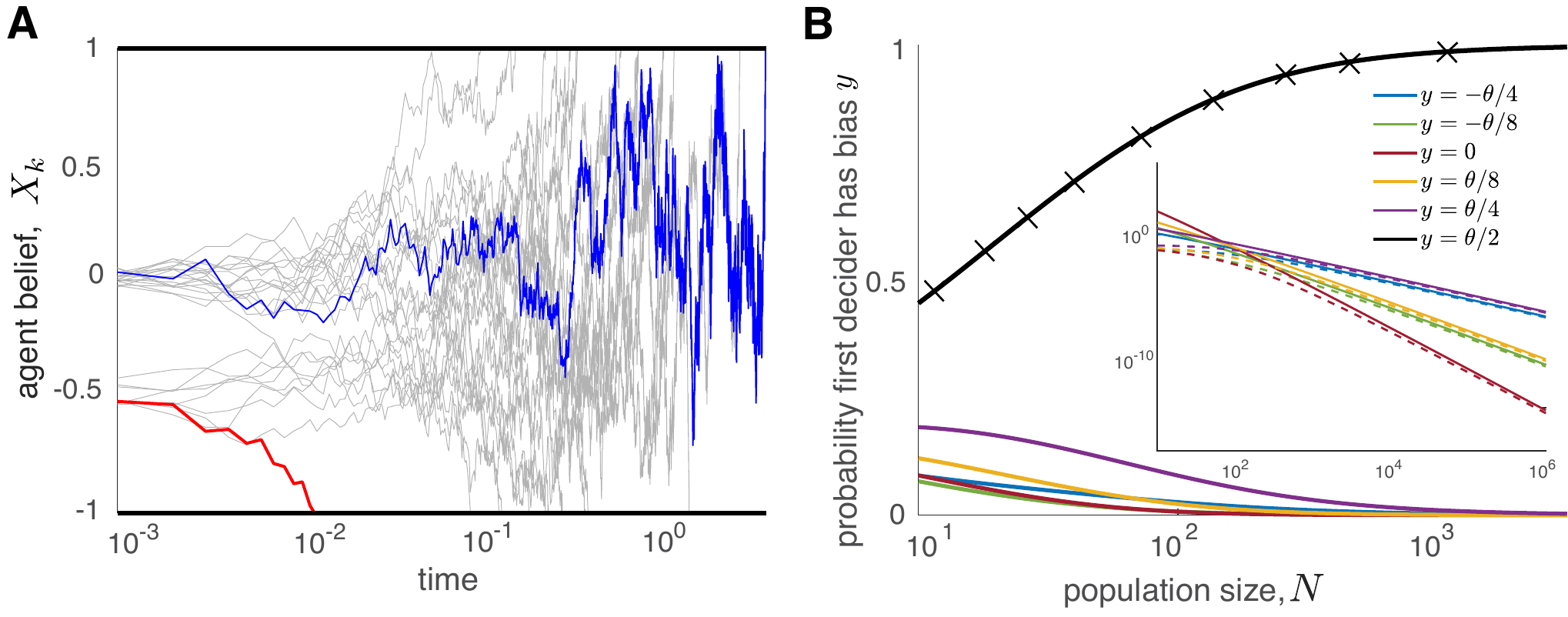}
\caption{Initial bias determines the choice of early deciders. (A)~Evolution of beliefs of $N = 10^4$ agents who each have even odds of initially being unbiased  or biased  ($P(X_j(0) = x_i) =0.5$, $x_i=0,-0.5$). The first agent (red) decides according to their initial bias, and makes the wrong decision at $T_1 \approx 0.01$. The last agent (blue) decides correctly at $T_{10000} \approx 10$. (B)~Probability that the agent with the largest initial bias decides first as a function of population size, $N$. Solid curves were determined by numerical quadrature
(Eq.~(\textcolor{blue}{S4}))
with initial biases assigned with uniform probability from values listed in the legend; black crosses denote results of a stochastic simulation averaged over $10^6$ trials. Inset: Log-log plot of the same results with dashed curves showing the asymptotic results in Eq.~\eqref{E:decay}.  Throughout, agents use identical thresholds $\pm \theta = \pm 1$, drift $\mu = 1$, and diffusivity $D = 1$.}
\label{fig1}
\end{figure}

For a fixed number of initial beliefs, $I$, the first agent to decide in a large group is the one with the largest initial bias (Fig.~\ref{fig1}A), in the sense that
\begin{align}\label{want}
P(X_{n(1)}(0)=x_{0})\to1\quad\text{as }N\to\infty.
\end{align}
More precisely, in the Supplemental Material (SM) we show that
\begin{align} \label{E:decay}
P(X_{n(1)}(0)=x_{i})\sim \eta_{i} (\ln N)^{(\beta_{i}-1)/2}N^{1-\beta_{i}}
\end{align}
as $N\to\infty$ for each $i\neq 0$, where 
\begin{align*}
\beta_{i}
&=(L_{i}/L_{0})^{2}>1, \\
\eta_{i}
&=\frac{q_{i}}{q_{0}^{\beta_{i}}}\exp\Big(\frac{\sqrt{\beta_{i}}}{2D}\big(\mu_i L_{0}-\mu_{0}L_i\big)\Big)\sqrt{\beta_{i}\pi^{\beta_{i}-1}}\Gamma(\beta_{i})
>0,
\end{align*}
and $\mu_i= \pm \mu$ if $x_i \gtrless 0$.  The same statement holds if $n(1)$ is replaced by $n(j)$ in Eq.~\eqref{E:decay}, but with a change in the prefactor, $\eta_i$ (See SM).
Thus, the probability that the
first decision is {\em not} made by the agent with the most extreme initial belief decreases as a negative power of the population size $N$ (Fig.~\ref{fig1}B). 
The approximation given by Eq.~\eqref{E:decay} is in excellent agreement with 
the true probabilities when $N \gtrapprox 10^3$ (See Fig~\ref{fig1}B inset).
Moreover, the probability that the agents with the most extreme initial beliefs make the 
first decision is close to unity already for $N \approx 100$ when initial beliefs are well separated and drift is not too strong.

\begin{figure}[t!]
\includegraphics[width = \columnwidth]{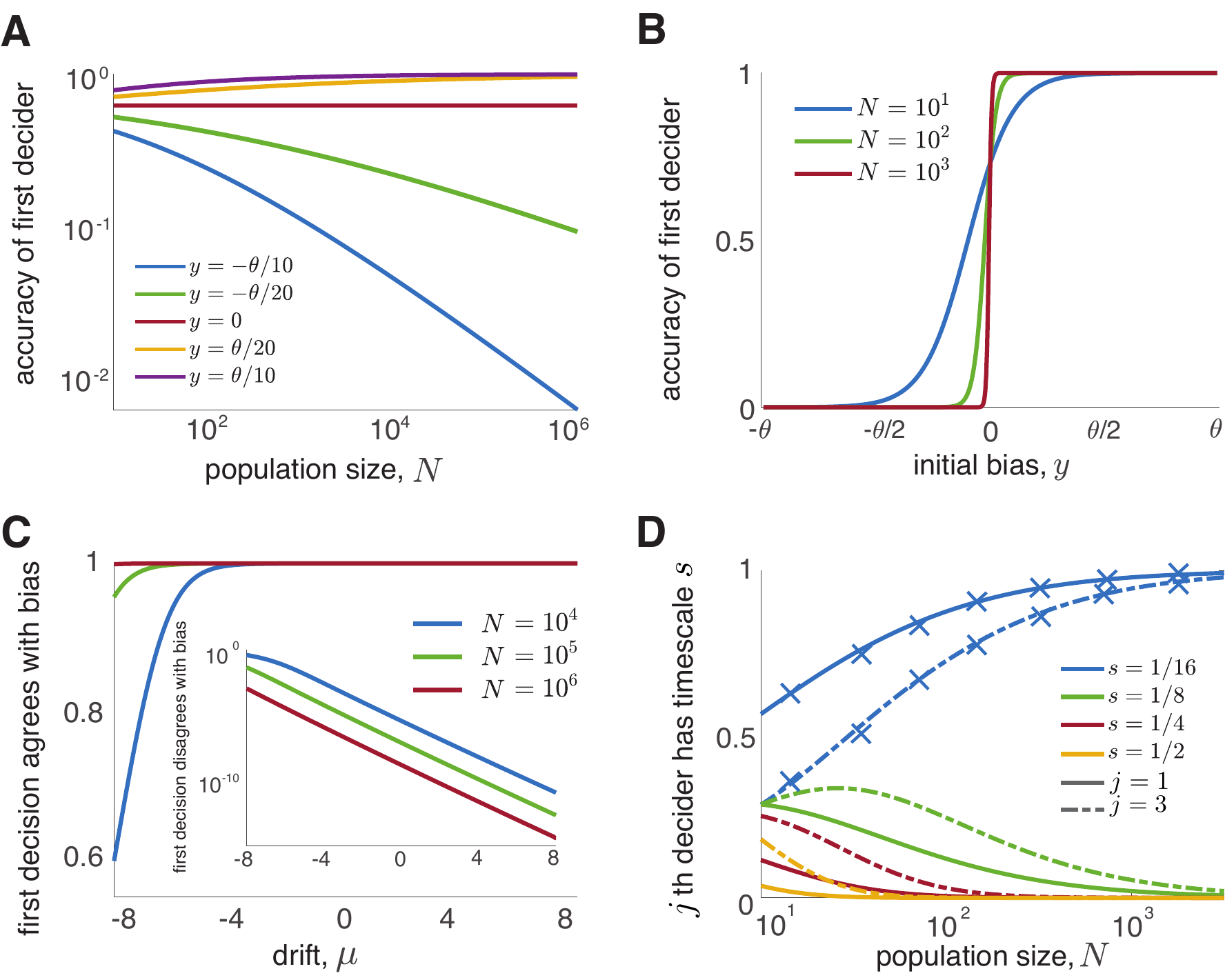}
\caption{First decider accuracy is determined by its initial bias. (A) The accuracy of the first decider as a function of population size, $N,$ for different initial biases, $y,$ obtained by quadrature. Curves are ordered by the proximity of the initial bias $y$ of the first decider to the correct threshold $+ \theta$. The drift, and hence the correct decision, are positive. (B) Under the same assumptions a small deviation from an unbiased initial belief strongly affects the probability of a correct first decision when  $N$ is large. (C) Drift weakly affects the first decision in populations with biased agents ($y = \theta/4$ here) when $N$ is large. See SM for decision polarity formulas. (D) In large populations in which all agents have the same initial bias, $y=\theta/2,$ but different diffusivities, early deciders (here first and third) have the shortest diffusive timescale. X's represent averages of stochastic simulations  over $10^6$ trials. }
\label{fig2}
\end{figure}

The choice of the fastest decider agrees with their initial bias:
e.g., if $\theta$ is the threshold closest to the most extreme initial belief, $x_{0}$, then $P(X_{n(1)}(T_1)= \theta)\to1 \text{ as }N\to\infty$ (See Fig.~\ref{fig2}A,B).  
Similar results hold when initial beliefs are drawn from a continuous distribution (See SM and next section). Thus, although all agents behave rationally, early decisions of biased agents tend to be less accurate~\cite{Karamched20,stickler2023impact}.


In contrast, the probability that a single 
agent - or one chosen randomly without regard to decision order -  decides incorrectly  can be made arbitrarily small by increasing the drift or threshold~\cite{Bogacz2006}.  In large populations with biased agents, drift and diffusion impact the probability of the first decision only through the prefactor in Eq.~(\ref{E:decay}),
$\eta_i$, and thus decrease in importance as population size diverges. If even a small proportion of a large population holds an initial bias, early
decisions are determined by the most extreme bias (Fig.~\ref{fig2}B) regardless of the drift (Fig.~\ref{fig2}C).
On the other hand, if all deciders are initially unbiased ($X_i(0) =0$ for all $i$), the probability the first decider
makes a correct choice is $(1 + \exp(-\mu \theta/D))^{-1}$~\cite{Bogacz2006}.

\paragraph{Heterogeneous population and continuous distribution of initial biases.} 
While we can obtain the most precise asymptotic results in the homogeneous case, our conclusions extend to
populations of agents with heterogeneous distributions of initial biases, drifts, diffusivities, and thresholds.  We again assume that each agent again starts with one of finitely many initial beliefs, $X_i(0) \in \{x_0,x_1, \ldots, x_{I-1}\}$ with drift and diffusivity sampled from a finite set of fixed size.  For each agent we define the diffusive timescale, 
\begin{align} \label{E:distance}
    S_i=\frac{L_i^2}{4D_i}>0.
\end{align}
By assumption, the
timescales $S_i$ follow a discrete distribution $
    \P(S=s_i)>0 $ with support on a finite set
$
    0<s_0 \le s_1\le s_2\le s_3\dots\le s_J,
$
and $S_{n(j)}$ refers to the timescale of the $j^{\rm th}$ agent to decide (See Fig.~\ref{fig2}D). We denote by $s$ the diffusive timescale of a generic agent.

In large populations, early deciders are those with the shortest diffusive timescales. In particular, we show in the SM that for every $\eps>0$ and \emph{fixed} $j\ge 1$, 
\begin{align}\label{simpleall}
    N^{1-s_1/s_0-\eps}
    \ll\P(S_{n(j)}>s_0)
    \ll N^{1-s_1/s_0+\eps}\quad\text{as }N\to\infty,
\end{align}
where we use the notation $f\ll g$ to mean $\lim_{N \rightarrow \infty} f/g=0$. We can thus conclude that
$
N^{1-s_1/s_0-\eps}
    =o\big(\P(S_{n(j)}>s_0)\big)$ and 
$
    \P(S_{n(j)}>s_0)
    =o\big(N^{1-s_1/s_0+\eps}\big)$ as $N\to\infty$.

These results agree with our earlier conclusion: If all agents share  the same diffusivity, then the fastest deciders are the agents who start closest to their decision thresholds. 
This is true regardless of the quality of the evidence they receive.
 Diffusivity can reduce the effective distance to the threshold according to Eq.~\eqref{E:distance}.  Thus, the fastest deciders are either those with the most 
extreme initial biases or those with the noisiest integration process, regardless of 
the drift, $\mu_i$. Indeed, how we model drift does not impact these conclusions, and they hold even if we model the evolution of beliefs as an 
Ornstein-Uhlenbeck process, as is frequently done in the psychophysics 
literature~\cite{busemeyer1993decision,veliz16}.  



\begin{figure}
\includegraphics[width = \columnwidth]{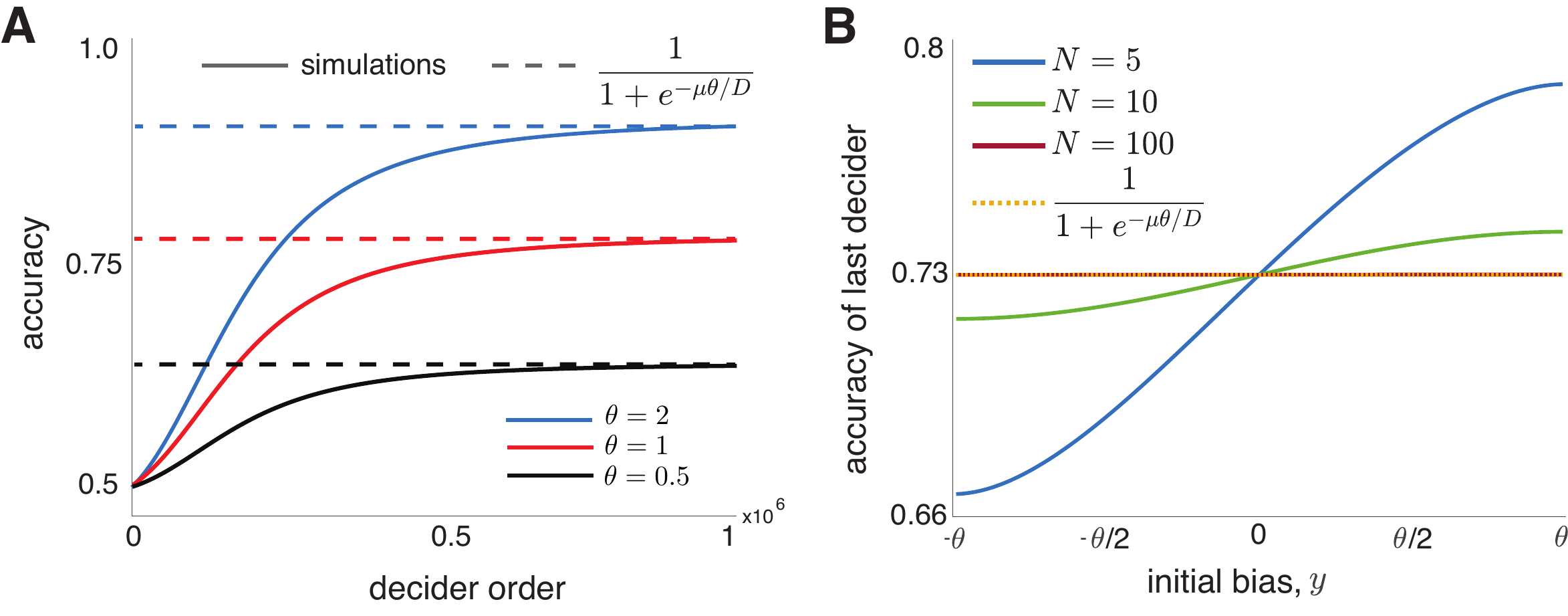}
\caption{Late deciders make choices as if they held no initial bias. (A) For large $N$, decision accuracy monotonically increases with decision order. The accuracy of late deciders approaches the accuracy of a single, initially unbiased agent. Here, all agents have initial bias $\theta/3$, and on each trial, $\mathbb{P}(H = H^+) = 0.5$. (B) In large groups even a large initial bias has no impact on the decision of later agents. Here, initial biases are sampled with uniform probability from $(-\theta,\theta)$.}
\label{fig3}
\end{figure}

\paragraph{Late deciders make decisions as if initially unbiased.}
We expect in large populations the inaccuracy of early deciders to be balanced by higher accuracy of late deciders~\cite{stickler2023impact}. Thus, we next determine the probability that the last agent to decide makes a correct decision.
In the SM we show that this probability has an intuitive form,
\begin{align} \label{E:last}
    \P(X_{n(N)}(T_N) = \theta) 
    \to\int_{-\theta}^{\theta} p_{\theta}(x)q(x)\,\dd x\quad\text{as }N\to\infty.
\end{align}
Here $p_{\theta}(x)$ is the probability that a single agent with initial bias 
$X(0) = x$ makes a correct decision, and $q(x)$ is the quasi-steady state distribution~\footnote{This is the distribution of beliefs conditioned on the absence of a decision after a long time.}
of beliefs evolving according to Eq.~\eqref{E:accumulation}. Thus the decision
of the last decider is made as if they forget their actual initial bias and instead sample an initial belief from the quasi-stationary distribution, $q(x)$. 

Eq.~\eqref{E:last} is general and can be extended to arbitrary domains. When
applied to the drift-diffusion process with decision boundaries at $\pm \theta$ we
show in the SM that $\P(X_{n(N)}(T_N) = \theta) 
    \to (1 + \exp(-\mu \theta/D))^{-1}$ as $N\to\infty$, which is the probability that
    a single, initially unbiased decider makes a correct decision (See Fig.~\ref{fig3}A)~\cite{Bogacz2006}. 
    Thus, the last decider forgets their initial bias and makes decisions based only on the accumulated evidence. The probability that an agent with a large initial bias makes a late decision is small. But should this happen, the initial bias will have little impact on their decision (See Fig.~\ref{fig3}B).

\begin{figure}
\includegraphics[width = 0.5\columnwidth]{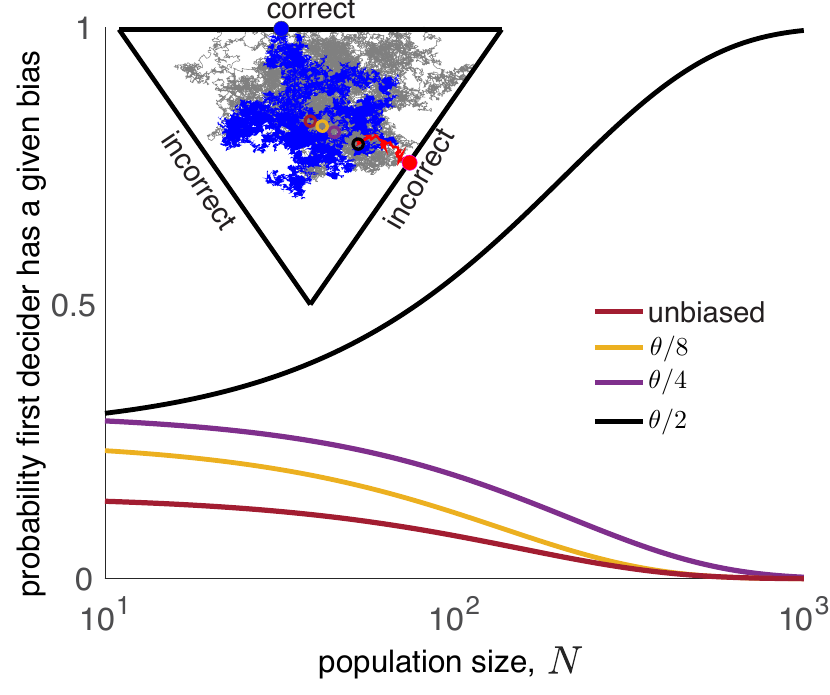}
\caption{Bias impacts multi-alternative and two-alternative decisions similarly in large groups. Beliefs about three options evolve on an equilateral triangle. Here, $\theta$ is the closest distance from the center of the triangle (burgundy ring) to the boundary. The initial bias is the distance from the triangle center to the initial belief, $X_i(0)$. As $N$ increases, the probability that the most biased agent chooses first grows. Curves are computed by averaging $10^6$ stochastic simulations. Inset:~Sample trajectories from a trial with biases sampled with equal probability from $\{\theta/2, \theta/4, \theta/8\}$. The first agent to decide (red) has the largest initial  bias. The belief of the last decider (blue) explores the space before reaching a threshold.}
\label{fig4}
\end{figure}

\paragraph{Extension to multiple alternatives.}
We can extend these results to decisions between $k$ alternatives.  
 Eq.~\eqref{E:accumulation} again describes the evolution of
beliefs, but now $X_i(t), \mu_i \in \R^{k-1}$ and
$W_i$ is a vector of independent Wiener processes~\footnote{We interpret beliefs as log-likelihood ratios. Therefore, with $k$ alternatives beliefs evolve in $k-1$ dimensions.}. Each belief
evolves on a domain, $\Omega \subset \mathbb{R}^{k-1}$, with $k$ boundaries~\cite{mcmillen2006dynamics}, 
each associated with one of the alternatives.
Agent $i$ chooses alternative $j$ if its belief, $X_i(t)$, crosses the associated boundary first. 
The boundaries that lead to the best decisions
are difficult to find analytically~\cite{tajima2019optimal}, but their exact shape is immaterial for our result.

In the SM we show that Eq.~\eqref{E:decay} holds for general domains in arbitrary dimensions (See Fig.~\ref{fig4}).
We therefore reach our earlier conclusions: In large homogeneous populations, the agents holding the most extreme initial beliefs make the first decisions, and their choices are consistent with their initial biases.
Our conclusions about the late decisions also carry over to agents facing multiple 
choices: The natural extension of Eq.~\eqref{E:last} holds with $q(x)$ the quasi-stationary 
distribution on $\Omega$.  The last decider makes a choice as if it sampled its initial belief from this quasi-stationary distribution.

\paragraph{Discussion}
Our decisions are often influenced by information we obtained previously and predilections we develop. 
In drift-diffusion models, prior evidence and initial inclinations are often represented by a shift in the initial state. 
We have shown that initial biases determine early decisions and have a diminishing impact on later decisions. 

An agent unaware of the order of their decision would believe this decision was made according to the evidence the agent accumulated and that the accuracy of their choice is determined only by the decision threshold~\cite{Bogacz2006}. Though early decisions are not always necessarily less accurate~\cite{karamched2020heterogeneity}, our work identifies a clear case in which hasty choices tend to be the most unreliable. Our findings also suggest a means of weighting choices of biased agents according to decision order in a large group when formulating collective decisions~\cite{marshall2017individual}. 
However, in social groups the exchange of social information between agents~\cite{Acemoglu2011,mossel2014opinion} or 
correlations in the evidence~\cite{stickler2023impact} will affect these results.

Ramping activity of individual neurons during decision making has been observed across the brain~\cite{shadlen1996,kiani2008bounded} (although see~\cite{latimer2015single}). Such dynamics may reflect the underlying evidence accumulation process preceding a decision and is often modeled by 
a drift-diffusion process. Decisions are thought to be triggered by the elevated activity of sufficiently many choice-related neurons~\cite{wang2012neural}. Our results suggest that in large neural populations decisions reflect the most extreme initial neural states, rather than the accumulated evidence, if the activity is uncorrelated. Since neural activity is often correlated~\cite{cohen2011measuring}, the effect of such biases  could be tempered.


While we have interpreted our results in the context of social decision theory, they apply more generally to independently evolving drift-diffusion processes on bounded domains~\cite{linn2022extreme}: In large populations early threshold crossings reflect only the initial states, while late crossings are independent of initial states and reflect the quasi-stationary distribution. Hence, early crossings reflect initial biases providing fast reactions needed for deadlined biophysical processes~\cite{grebenkov2020single}. If time allows, quorum sensing processes that weight passages by order could be used~\cite{koriat2012two}. Thus, our theory shows how initial biases can be used to implement population level tradeoffs between speed and accuracy.

\section{Acknowledgements}
ZPK and KJ were supported by NSF DMS-2207700. ZPK was supported by NIH BRAIN 1R01EB029847-01. SDL was supported by NSF DMS-2325258 and NSF DMS-1944574. KJ was supported by NSF-DBI-1707400 and NIH RF1MH130416. SL was supported by NSF Grant No. 2139322.

\newpage

\begin{center}
\Large{Supplemental Material}
\end{center}

\section{Mathematical preliminaries}\label{prelim}

Suppose $\{(\tau_{n},Z_{n})\}_{n\ge1}$ is an independent and identically distributed (iid) sequence of realizations of the pair of (possibly correlated) random variables $(\tau,Z)$. We have in mind that $\tau$ is the decision time (or first passage time (FPT)) of some decider whose stochastic evolution of beliefs is denoted by $\{X(t)\}_{t\ge0}$ and $Z$ is a vector containing information about this decider, such as their random initial position, drift, diffusivity, and decision made. 
Define the cumulative distribution function (CDF) of $\tau$,
\begin{align*}
F(t)
:=\P(\tau\le t).
\end{align*}
Further, for any event $E$ that is in the $\sigma$-algebra generated by $Z$, define
\begin{align*}
F_{E}(t)
:=\P(\tau\le t\cap E).
\end{align*}
In words, $E$ is any event for which we can know whether or not it occurred by knowing $Z$. For example, we are interested in events $E$ like $E=\{X(0)=\theta/2\}$, $E=\{X(0)\le 0\}$, $E=\{X(\tau)=\theta\}$, etc.

For a given $N\ge1$, let $n(j)\in\{1,\dots,N\}$ denote the (random) index of the $j$th fastest decider out of the first $N$ deciders to make a decision. That is, suppose we order the first $N$ FPTs (or first decision times),
\begin{align*}
T_{1,N}\le T_{2,N}\le \cdots\le T_{N-1,N}\le T_{N,N},
\end{align*}
where $T_{j,N}$ denotes the $j$th fastest FPT,
\begin{align}\label{Tkn}
T_{j,N}
:=\min\big\{\{\tau_{1},\dots,\tau_{N}\}\backslash\cup_{i=1}^{j-1}\{T_{i,N}\}\big\},\quad j\in\{1,\dots,N\}.
\end{align}
Then $n(j)$ is such that
\begin{align}\label{njn}
\tau_{n(j)}
=T_{j,N}.
\end{align}
In the examples of interest, the FPTs, $\tau,$ have continuous probability distributions (i.e.\ $F(t)$ is a continuous function) so that the event $\tau_{n^{*}}=\tau_{n'}<\infty$ for $n^{*}\neq n'$ has probability zero so there is no ambiguity in Eq.~\eqref{njn}.

Since we have the sequence $\{(\tau_{n},Z_{n})\}_{n\ge1}$, we denote $E_{n}$ the event $E$ as it pertains to the $n$th element in the sequence $\{(\tau_{n},Z_{n})\}_{n\ge1}$. For example, if $E=\{X(0)=\theta/2\}$, then $E_{n}=\{X_{n}(0)=\theta/2\}$. Similarly, $E_{n(j)}$ is the event $E$ as it pertains to $Z_{n(j)}$.

Throughout the Supplemental Material, we use the notation $\int f(t)\,\dd g(t)$ to denote the Riemann-Stieltjes integral of a function $f$ with respect to a function $g$.

\begin{proposition}\label{genint}
For any $j\in\{1,2,\dots,N\}$ (denoting an agent by the order $j$ of their decision), we have that
\begin{align}\label{jj}
\P(E_{n(j)})
=j{N\choose j}\int_{0}^{\infty}[F(t)]^{j-1}[1-F(t)]^{N-j}\,\dd F_{E}(t).
\end{align}
\end{proposition}

In the case $j=1$ (i.e.\ the fastest decider), Proposition~\ref{genint} implies
\begin{align}\label{j1}
\P(E_{n(1)})
=N\int_{0}^{\infty}[1-F(t)]^{N-1}\,\dd F_{E}(t).
\end{align}
Since $1-F$ is a decreasing function, Eq.~\eqref{j1} implies that the short-time behavior of $F$ and $F_{E}$ determine the large $N$ behavior of $\P(E_{n(1)})$. More generally, Proposition~\ref{genint} implies that the short-time behavior of $F$ and $F_{E}$ determine the large $N$ behavior of $\P(E_{n(j)})$ for $1\le j\ll N$.

In the case $j=N$ (i.e.\ the slowest decider), Proposition~\ref{genint} implies
\begin{align}\label{jN}
\P(E_{n(N)})
=N\int_{0}^{\infty}[F(t)]^{N-1}\,\dd F_{E}(t).
\end{align}
Since $F$ is an increasing function, Eq.~\eqref{jN} implies that the large-time behavior of $F$ and $F_{E}$ determine the large $N$ behavior of $\P(E_{n(N)})$. More generally, Proposition~\ref{genint} implies that the large-time behavior of $F$ and $F_{E}$ determine the large $N$ behavior of $\P(E_{n(N-j)})$ for $1\ll N-j$.

\section{Some integral asymptotics}

The following proposition is useful for estimating the large $N$ behavior of some integrals of the form in Eq.~\eqref{jj} and was proved in \cite{linn2022extreme} (See Proposition~2 in \cite{linn2022extreme}). Throughout the Supplemental Material, ``$f\sim g$'' denotes $f/g\to1$ (e.g., as $N\to \infty$ or as $t \to 0$).

\begin{proposition}\label{p1}
Assume $C_{+}>C>0$, $A>0$, and $p,q\in\R$. Then there exists a $\delta_{0}>0$ so that for all $\delta\in(0,\delta_{0}]$, we have
\begin{align*}
&\int_{0}^{\delta}t^{q-2}e^{-C_{+}/t}\big(1-At^{p}e^{-C/t}\big)^{N-1}\,\dd t
\sim
{{\eta}}
(\ln N)^{p{\beta}-q}
N^{-{\beta}}
\quad\text{as }N\to\infty,
\end{align*}
where
\begin{align*}
{\beta}
&=C_{+}/C>1,\quad
{{\eta}}
=C^{q-1}(AC^{p})^{-\beta}\Gamma(\beta)>0,
\end{align*}
and $\Gamma({\beta}):=\int_{0}^{\infty}z^{{\beta}-1}e^{-z}\,\dd z$ denotes the gamma function.
\end{proposition}

The following result estimates integrals of the form in Eq.~\eqref{jj} for $1\le j\ll N$ assuming that $F(t)$ and $F_{+}(t)$ have short-time $t$ behavior that is characteristic of diffusion. 

\begin{theorem}\label{detailprobj}
Assume $F(t)$ and $F_{+}(t)$ are bounded, nondecreasing, continuous from the right, and satisfy
\begin{align}
F(t)
&\sim At^{p}e^{-C_{0}/t}\quad\text{as }t\to0^+,\label{s1}\\
F_{+}(t)
&\sim Bt^{q}e^{-C_{+}/t}\quad\text{as }t\to0^+,\label{s2}
\end{align}
where $C_{+}>C_{0}>0$, $A>0$, $B>0$, and $p,q\in\R$. Then for any fixed integer $j\ge1$, we have
\begin{align*}
j{N\choose j}\int_{0}^{\infty}[F(t)]^{j-1}[1-F(t)]^{N-j}\,\dd F_{+}(t)
\sim
{{\eta(j)}}
(\ln N)^{p{\beta}-q}
N^{1-{\beta}}
\quad\text{as }N\to\infty,
\end{align*}
where
\begin{align}\label{bk}
{\beta}
&:=C_{+}/C_{0}>1,\quad
{{\eta(j)}}
:=B(C_{0})^{q-p\beta}A^{-\beta}\Gamma(j)\Gamma(\beta+j)
>0,
\end{align}
and $\Gamma(x):=\int_{0}^{\infty}z^{x-1}e^{-z}\,\dd z$ denotes the gamma function.
\end{theorem}

Notice that the asymptotic behavior found in Theorem~\ref{detailprobj} as $N \to \infty$ is independent of $j\ge1$, except for the constant prefactor $\eta(j)$. Further, this prefactor is an increasing function of $j$ and satisfies
\begin{align*}
    \eta(j)=\frac{(j-1)!\Gamma(\beta+j)}{\Gamma(\beta+1)}\eta(1),\quad j\ge1.
\end{align*}

The asymptotic behavior in Eq.~\eqref{s1}-\eqref{s2} is typical for diffusion, but computing the prefactors $A$ and $B$ and the powers $p$ and $q$ can be challenging \cite{lawley2020dist}. Indeed, these constants depend on the details of the system (e.g.,~drift, space dimension, geometry of the domain, etc.). However, the constants in the exponents $C_{0}$ and $C_{+}$ are more universal and can be obtained in a very general mathematical setting \cite{Lawley20b}. The following result yields estimates on the fastest deciders when we only know these constants, which is equivalent to knowing the short-time behavior of $F_{+}(t)$ and $F(t)$ on a logarithmic scale.

\begin{theorem} \label{logprob}
Assume $F(t)$ and $F_{+}(t)$ are bounded, nondecreasing, continuous from the right, and satisfy
\begin{align}\label{lb}
\lim_{t\to0^+}t\ln F(t)
&= -C_{0}<0,\quad
\lim_{t\to0^+}t\ln F_{+}(t)
\le -C_{+}<0,
\end{align}
where $C_{+}>C_{0}>0$. Then for every $\eps>0$, 
\begin{align}\label{ub}
j{N\choose j}\int_{0}^{\infty}[F(t)]^{j-1}[1-F(t)]^{N-j}\,\dd F_{+}(t)
=o(N^{1-{\beta}+\eps})
\quad\text{as }N\to\infty,
\end{align}
where
\begin{align*}
{\beta}
&:=C_{+}/C_{0}>1.
\end{align*}
If, in addition, we assume that
\begin{align}\label{further}
    \lim_{t\to0^+}t\ln F_{+}(t)
= -C_{+}<0,
\end{align}
then for every $\eps>0$,
\begin{align*}
    N^{1-\beta-\eps}
    =o\bigg(j{N\choose j}\int_{0}^{\infty}[F(t)]^{j-1}[1-F(t)]^{N-j}\,\dd F_{+}(t)\bigg)\quad\text{as }N\to\infty.
\end{align*}
\end{theorem}

The following result estimates integrals of the form in Eq.~\eqref{jj} for $1\ll N-j\le N$ assuming that $F(t)$ and $f_i(t)=F_{i}'(t)$ have large-time $t$ behavior that is characteristic of diffusion in a bounded domain.

\begin{theorem}\label{estslow}
    Assume $F(t)\in[0,1)$ is continuous and nondecreasing and $f_i(t)$ is continuous and bounded and
    \begin{align*}
        F(t)
        &=1-ce^{-\lambda t}+\hot\quad\text{as }t\to\infty,\\
        f_i(t)
        &=\lambda c_i e^{-\lambda t}+\hot\quad\text{as }t\to\infty,
    \end{align*}
    where $\lambda>0$, $c>0$, $c_i>0$. 
    Then for any fixed $j\ge0$, we have that
    \begin{align*}
        (N-j){N\choose N-j}\int_0^\infty[F(t)]^{N-j-1}[1-F(t)]^j f_i(t)\,\dd t
        \to\frac{c_i}{c}\quad\text{as }N\to\infty.
    \end{align*}
\end{theorem}

\section{Proof of Eq.~(3) in main text}\label{pe3}

We now apply Theorem~\ref{detailprobj} to obtain Eq.~(3) in the main text. Suppose the belief of each agent evolves independently according to the following stochastic differential equation (SDE),
\begin{align}\label{simplesde}
\dd X
=\mu\,\dd t+\sqrt{2D}\,\dd W,
\end{align}
where $\mu\in\R$ is a constant drift, $D>0$ is a constant diffusivity, and $W=\{W(t)\}_{t\ge0}$ is a standard Brownian motion. Define the FPT,
\begin{align*}
\tau
:=\inf\{t>0:X(t)\notin(-\theta,\theta)\},
\end{align*}
for some threshold $\theta>0$. Assume that the initial distribution $\P(X(0)=x_{i})$ of each agent is a sum of Dirac masses at a finite set of points $\{x_{0},x_1,\dots,x_{I-1}\}$,
\begin{align*}
\P(X(0)=x)
=\begin{cases}
q_{i} & \text{if }x=x_{i}\text{ for some }i\in\{0,1,\dots,I-1\},\\
0 & \text{if }x\notin\cup_{i=0}^{I-1}x_{i}.
\end{cases},
\end{align*}


Letting $F_{i}(t) \equiv F_{X(0) = x_i}(t) =\P(\tau\le t\cap X(0)=x_{i})$, we have that \cite{linn2022extreme} 
\begin{align}\label{eq:asymp1ddrift}
F_{i}(t)
\sim q_{i}A_{i}t^{1/2}e^{-C_{i}/t}\quad\text{as }t\to0^+,
\end{align}
where
\begin{align*}
C_{i}
=\frac{(L_{i})^{2}}{4D},
\end{align*}
and
\begin{align*}
    A_{i}
=\begin{cases}
    \exp\big(\frac{-\mu L_{i}}{2D}\big)\sqrt{\frac{4D}{\pi (L_{i})^{2}}} & \text{if }x_i<0\\
    \exp\big(\frac{\mu L_{i}}{2D}\big)\sqrt{\frac{4D}{\pi (L_{i})^{2}}} & \text{if }x_i>0\\
    \big[\exp\big(\frac{-\mu L_{i}}{2D}\big)+\exp\big(\frac{\mu L_{i}}{2D}\big)\big]\sqrt{\frac{4D}{\pi (L_{i})^{2}}} & \text{if }x_i=0,
\end{cases}
\end{align*}
where $L_{i}$ is the distance to the closest threshold from $x_{i}$,
\begin{align*}
L_{i}=\min\{\theta-x_{i},x_{i}+\theta\}.
\end{align*}
Further, we assume ${0}\in\{0,1,\dots,I-1\}$ is the index of the unique starting location closest to a threshold
\begin{align*}
L_{{0}}
=\min\{L_0,L_{1},\dots,L_{I-1}\}
<L_i\quad\text{if }i\neq {0},
\end{align*}
then
\begin{align*}
F(t)
\sim F_{{0}}(t)\quad\text{as }t\to0^+.
\end{align*}

We claim that
\begin{align}\label{want}
\P(X_{n(1)}(0)=x_{{0}})\to1\quad\text{as }N\to\infty,
\end{align}
Thus, when $N$ is large the first decider out of many deciders is always the one with the most extreme initial bias. Using the integral representation in Proposition~\ref{genint} and applying Theorem~\ref{detailprobj} yields
\begin{align*}
\P(X_{n(1)}(0)=x_{i})\sim \eta_{i}(1) (\ln N)^{(\beta_{i}-1)/2}N^{1-\beta_{i}}\quad\text{as }N\to\infty\quad\text{for each }i\neq {0},
\end{align*}
where 
\begin{align*}
\beta_{i}
&=(L_{i}/L_{{0}})^{2}>1,
\end{align*}
and
\begin{align*}
    \eta_{i}(1)
&=
\begin{cases}
\displaystyle\frac{q_i}{q_{{0}}^{\beta_i}}\sqrt{\frac{\pi^{\beta_{i}-1}}{\beta_{i}}}\Gamma(\beta_{i}+1)\exp\big(\frac{\sqrt{\beta_{i}}}{2D}\big(\mu_i L_{{0}}-\mu_{{0}}L_i\big)\big) & \text{if }x_i\neq0,\\
\displaystyle \frac{q_i}{q_{{0}}^{\beta_i}}\sqrt{\frac{\pi^{\beta_{i}-1}}{\beta_{i}}}\Gamma(\beta_{i}+1)\Big[\exp\big(\frac{\sqrt{\beta_{i}}}{2D}\big(\mu_i L_{{0}}-\mu_{{0}}L_i\big)\big)+\exp\big(\frac{\sqrt{\beta_{i}}}{2D}\big(-\mu_i L_{{0}}-\mu_{{0}}L_i\big)\big)\Big] & \text{if }x_i=0,
\end{cases}
\end{align*}
where $\mu_i = \pm \mu$ if $x_i \gtrless 0$.

\section{First decision agrees with initial bias}

The analysis above shows that the first agent to decide in a large group has the most extreme initial bias. We now show the intuitive result that this first decider's decision agrees with their initial bias. Without loss of generality, assume that the most extreme initial bias is negative, $x_{{0}}<0$. Letting $F_+(t)=\P(\tau\le t\cap X(\tau)=+\theta)$, we have
\begin{align*}
    F_+(t)
    &=\sum_i \P(\tau\le t\cap X(\tau)=+\theta\,|\,X(0)=x_i)q_i\\
    &\sim\P(\tau\le t\cap X(\tau)=+\theta\,|\,X(0)=x_{i^+})q_{i^+}\\
    &\sim q_{i^+}A_{i^+}t^p_{i^+} e^{-C_{i^+}/t}\quad\text{as }t\to0^+,
\end{align*}
where $i^+\in\{1,\dots,I\}$ is the index of the starting location closest to $+\theta$. Using the integral representation in Proposition~\ref{genint} and applying Theorem~\ref{detailprobj} yields
\begin{align*}
\P(X_{n(1)}(\tau)=+\theta)\sim \eta_{i^+}^{(1)} (\ln N)^{(\beta_{i^+}-1)/2}N^{1-\beta_{i^+}}\quad\text{as }N\to\infty. 
\end{align*}

\section{Continuous initial belief distribution}

In Section~\ref{pe3}, we showed that the first of many deciders have the most extreme initial beliefs in the case that the population has a discrete initial belief distribution. We now generalize this calculation to the case that the deciders have a continuous initial belief distribution. In particular, suppose that the decider's initial belief (position) has a smooth probability density $\nu(x)$ with support $(a,b)$ with $-\theta<a<b<\theta$. Suppose that 
\begin{align*}
    \nu(x)
    &\sim (x-a)^{\alpha_a}\nu_a\quad\text{as }x\to a^+,\\
    \nu(x)
    &\sim (b-x)^{\alpha_b}\nu_b\quad\text{as }x\to b^-,
\end{align*}
where the coefficients are positive, $\nu_a>0$, $\nu_b>0$, and the powers ensure that $\nu$ is integrable, $\alpha_a>-1$, $\alpha_b>-1$. In light of \eqref{eq:asymp1ddrift}, suppose that
\begin{align*}
    \P(\tau\le t\,|\,X(0)=x)\sim A(x)t^pe^{-C(x)/t}\quad\text{as }t\to0^+,\quad\text{uniformly for all }x\in[a,b],
\end{align*}
where
\begin{align*}
    C(x)
    =(L(x))^2/(4D)>0,
    \quad L(x)=\min\{\theta-x,\theta+x\},
\end{align*}
and $A(x)>0$ for all $x\in[a,b]$.

It follows that
\begin{align*}
    F(t)
    =\P(\tau\le t)
    &=\int_a^b \P(\tau\le t\,|\,X(0)=x)\nu(x)\,\dd x\\
    &\sim t^p\int_a^b A(x)\nu(x)e^{-C(x)/t}\,\dd x\quad\text{as }t\to0^+.
\end{align*}
We thus need to estimate the small time $t$ asymptotics of the integral
\begin{align*}
    I
    :=\int_a^b A(x)\nu(x)e^{-C(x)/t}\,\dd x,
\end{align*}
which is an exercise in Laplace's method \cite{bender2013}. If $b>0$, then for any $\eps\in(0,b)$, we have
\begin{align*}
    \int_{0}^b A(x)\nu(x)e^{-C(x)/t}\,\dd x
    &\sim\int_{b-\eps}^b A(x)\nu(x)e^{-C(x)/t}\,\dd x\\
    &\sim A(b)e^{-C(b)/t}\nu_b\Gamma(\alpha_b+1)t^{\alpha_b+1}\quad\text{as }t\to0^+.
\end{align*}
Similarly, if $a<0$, then for any $\eps\in(0,|a|)$, we have
\begin{align*}
    \int_{a}^0 A(x)\nu(x)e^{-C(x)/t}\,\dd x
    &\sim\int_{a}^{a+\eps} A(x)\nu(x)e^{-C(x)/t}\,\dd x\\
    &\sim A(a)e^{-C(a)/t}\nu_a\Gamma(\alpha_a+1)t^{\alpha_a+1}\quad\text{as }t\to0^+.
\end{align*}
Putting this together, we have that if $b>|a|$, then
\begin{align*}
    F(t)
    \sim A(b)\nu_b\Gamma(\alpha_b+1)t^{p+\alpha_b+1}e^{-C(b)/t}\quad\text{as }t\to0^+,
\end{align*}
and similarly if $|a|>b$ or $|a|=b$.

With these estimates, we can apply Theorem~\ref{detailprobj} to obtain estimates that the fastest decider(s) have extreme initial beliefs. In particular, suppose we want to estimate
\begin{align*}
    \P(a+\eps<X_{n(1)}(0)<b-\eps)\quad\text{for some small $0<\eps\ll1$},
\end{align*}
which is the probability that the fastest decider does not have extreme initial beliefs. If we define the event
\begin{align*}
    E=\{a+\eps<X(0)<b-\eps\},
\end{align*}
then using the notation of Section~\ref{prelim}, we have that
\begin{align*}
    F_E(t)
    :=\P(\tau\le t\cap E)
    &=\int_{a+\eps}^{b-\eps}\P(\tau\le t\,|\,X(0)=x)\nu(x)\,\dd x\\
    &\sim t^p\int_{a+\eps}^{b-\eps} A(x)\nu(x)e^{-C(x)/t}\,\dd x\quad\text{as }t\to0^+,
\end{align*}
which can be estimated as above using Laplace's method \cite{bender2013}. In particular, if $b>|a|$, then 
\begin{align*}
    F_E(t)
    \sim A(b-\eps)\nu(b-\eps)t^{p+1}e^{-C(b-\eps)/t}\quad\text{as }t\to0^+,
\end{align*}
assuming $\nu(b-\eps)>0$, and similarly if $|a|>b$ or $|a|=b$. With this short-time behavior of $F_E(t)$, we can then plug this into Theorem~\ref{detailprobj} to show that the first deciders have the most extreme initial beliefs.

\section{Heterogeneous population with multiple alternatives}

We next consider the generalized case where the beliefs of the agents in the population evolve as processes with (possibly space-dependent) drift, diffusion coefficient, initial position, and even domain (in their own arbitrary space dimension $d\ge1$). Suppose the belief of the $i$th decider evolves according to the following $d$-dimensional SDE, 
\begin{align}
\dd X_{i}
=\mu_{i}(X_{i})\,\dd t
+\sqrt{2D_{i}}\,\dd W_{i}, \label{E:multalt}
\end{align}
where $\mu_{i}:\R^{{d}}\to\R^{{d}}$ is a possibly space-dependent drift, $D_{i}>0$ is the diffusion coefficient, and $W(t)\in\R^{d}$ is a standard Brownian motion in $d$-dimensional space. 

 Let $L>0$ denote an agent's (random) shortest distance they must travel to hit the closest target and let $D>0$ denote the agent's diffusion coefficient. Define the random timescale
\begin{align*}
    S=\frac{L^2}{4D}>0.
\end{align*}
Suppose that $S$ has a discrete distribution on a finite set
\begin{align*}
    0<s_0< s_1< s_2< s_3\dots< s_I,
\end{align*}
where
\begin{align*}
    \P(S=s_i)=q_i>0,\quad\sum_{i=0}^I q_i=1.
\end{align*}
Since we have $N\ge1$ iid agents indexed from $n=1$ to $n=N$, we let $S_n$ denote the value of $S$ for the $n$th agent and $S_{n(j)}$ the value of $S$ for the $j$th fastest to decide. 

We have that \cite{Lawley20b}
\begin{align*}
    \lim_{t\to0^+}t\ln \P(\tau\le t)
&= -s_{0}<0,\quad
\lim_{t\to0^+}t\ln \P(\tau\le t\cap S=s_i)
= -s_i<0.
\end{align*}
Hence, Proposition~\ref{genint} and Theorem~\ref{logprob} imply that for any fixed $j\ge1$ and $i\in\{1,\dots,I\}$ and any $\eps>0$,
\begin{align}\label{simpleall}
    N^{1-s_i/s_0-\eps}
    \ll\P(S_{n(j)}=s_i)
    \ll N^{1-s_i/s_0^\eps}\quad\text{as }N\to\infty,
\end{align}
where we use the notation $f\ll g$ to mean $\lim f/g=0$. That is, in more traditional notation,
\begin{align*}
    N^{1-s_i/s_0-\eps}
    =o\big(\P(S_{n(j)}=s_i)\big)\quad\text{as }N\to\infty,\\
    \P(S_{n(j)}=s_i)
    =o\big(N^{1-s_i/s_0+\eps}\big)\quad\text{as }N\to\infty.
\end{align*}
In the special case that the agents all move in one space dimension and the drifts are spatially constant (but may differ between agents), we can get the constant and logarithmic prefactors on the decay of $\P(S_{n(j)}=s_i)$ as $N\to\infty$.

The result in Eq.~\eqref{simpleall} says that in a large population if all the agents have the same diffusion coefficient, then the fastest deciders started closest to their decision thresholds (targets). If we allow the diffusion coefficients to vary between agents, then \eqref{simpleall} implies that the fastest deciders started close to their decision thresholds and/or they had big diffusion coefficients.

\section{Slowest deciders}

Suppose the beliefs of the iid agents diffuse in some $d$-dimensional spatial domain $U\subset\R^d$ and can be absorbed at one of $m\ge2$ targets $V_0,\dots,V_{m-1}$ and let $\kappa\in\{0,\dots,m-1\}$ indicate which target the decider eventually hits.   Here, we will think of the $m$ targets as parts of the  $d-1$ dimensional boundary of the domain, and assume that hitting one of the targets triggers a decision.
Following \cite{pavliotis2014, madrid2020comp}, suppose the beliefs of the deciders evolve as stochastic process $\{X(t)\}_{t\ge0}$ that diffuse according to the SDE
\begin{align}\label{sde}
\dd X(t)
=-\nabla V(X(t))\,\dd t+\sqrt{2D}\,\dd W(t),
\end{align}
with reflecting boundary conditions. In Eq.~\eqref{sde}, the drift term is the gradient of a given potential, $V(x)$, and the noise term depends on the diffusion coefficient $D>0$ and a standard $d$-dimensional Brownian motion (Wiener process) $\{W(t)\}_{t\ge0}$. 
The survival probability conditioned on the initial position,
\begin{align*}
\S(x,t)
:=\P(\tau>t\,|\,X(0)=x),
\end{align*}
satisfies the backward Kolmogorov (also called backward Fokker-Planck) equation,
\begin{align}\label{backward}
\begin{split}
\tfrac{\partial}{\partial t}\S
&=\L \S,\quad x\in{U},\\
\S
&=0,\quad x\in\textup{targets},\\
\tfrac{\partial}{\partial\n}\S
&=0,\quad x\in\textup{reflecting boundary (if there is one)},\\
S
&=1,\quad t=0.
\end{split}
\end{align}
In Eq.~\eqref{backward}, the differential operator $\L$ is the generator (i.e.\ the backward operator) of Eq.~\eqref{sde},
\begin{align*}
\L
=-\nabla V(x)\cdot\nabla+D\Delta,
\end{align*}
and $\frac{\partial}{\partial\n}$ is the derivative with respect to the inward unit normal $\n:\partial{U}\to\R^{d}$.

Using the following weight function of Boltzmann form ,
\begin{align}\label{rho}
\rho(x)
:=\frac{e^{-V(x)/D}}{\int_{{U}}e^{-V(y)/D}\,\dd y},
\end{align}
one can check that the differential operator $\L$ is formally self-adjoint on the weighted space of square integrable functions (see, for example, Ref.~\cite{pavliotis2014}),
\begin{align*}
L_{\rho}^{2}({U})
:=\Big\{f:\int_{{U}}|f(x)|^{2}\rho(x)\,\dd x<\infty\Big\},
\end{align*}
using the boundary conditions in \eqref{backward} and the following weighted inner product,
\begin{align*}
(f,g)_{\rho}
:=(f,g\rho)
=\int_{{U}}f(x)g(x)\rho(x)\,\dd x,
\end{align*} 
where $(f,g)=\int_U f(x)g(x)\,\dd x$ denotes the standard $L^2$-inner product (i.e.\ with no weight function). Expanding the solution to \eqref{backward} yields,
\begin{align}\label{Seig}
\S(x,t)
=\sum_{n\ge1}(u_{n},1)_{\rho}e^{-\lambda_{n}t}u_{n}(x)
=\sum_{n\ge1}(u_{n},\rho)e^{-\lambda_{n}t}u_{n}(x),
\end{align}
where
\begin{align}\label{order}
0<\lambda_{1}<\lambda_{2}\le\dots,
\end{align}
denote the (necessarily positive) eigenvalues of $-\L$. The corresponding with eigenfunctions $\{u_{n}(x)\}_{n\ge1}$ satisfy the following time-independent equation,
\begin{align}\label{evalproblem}
\begin{split}
-\L u_{n}
&=\lambda_{n}u_{n},\quad x\in{U},
\end{split}
\end{align}
and identical boundary conditions as $\S$. Further, the eigenfunctions are orthogonal and are taken to be orthonormal, which means that
\begin{align}\label{orthonormal}
(u_{n},u_{m})_{\rho}
=\delta_{nm}\in\{0,1\},
\end{align}
where $\delta_{nm}$ denotes the Kronecker delta function (i.e.\ $\delta_{nn}=1$ and $\delta_{mn}=0$ if $n\neq m$). 

If the initial distribution of an agent has probability measure $\mu_{0}$,
\begin{align}\label{initial}
\P(X(0)\in B)=\mu_{0}(B)=\int_{B}1\,\dd \mu_{0}(x),\quad B\subset{U},
\end{align}
then the FPT $\tau$ has survival probability given by
\begin{align*}
S(t)
&:=\P(\tau>t\,|\,X(0)=_{\dist}\mu_{0})
=\int_{{U}}\S(x,t)\,\dd\mu_{0}(x),
\end{align*}
where the condition $X(0)=_{\dist}\mu_{0}$ in the conditional probability merely denotes that $X(0)$ has initial distribution given by $\mu_0$. 
Hence, we obtain the following representation for the survival probability,
\begin{align}\label{form}
S(t)
=\sum_{n\ge1}A_{n}e^{-\lambda_{n}t}
=\sum_{n\ge1}(u_n,\rho)(u_n,\dd\mu_0)e^{-\lambda_{n}t},
\end{align}
where the coefficients are given by the following integrals,
\begin{align}\label{As}
A_{n}
:=(u_{n},1)_{\rho}\int_{{U}}u_{n}(x)\,\dd\mu_{0}(x),
\quad n\ge1.
\end{align}

We have that the FPT $\tau$ to one of the targets has CDF
\begin{align*}
    F(t)
    =\P(\tau\le t)
    &=1-\P(\tau> t)\\
    &=1-\sum_{k\ge1} (u_k,\rho)(u_k,\dd\mu_0)e^{-\lambda_k t},
\end{align*}
If
\begin{align*}
    p_i(x)
    =\P(\kappa=i|X(0)=x),
\end{align*}
then
\begin{align*}
    F_i(t)
    :=\P(\tau\le t\cap\kappa=i)
    &=\P(\kappa=i)-\P(\tau>t\cap\kappa=i)\\
    &=\P(\kappa=i)-\sum_{k\ge1} (u_k,p_i\rho)(u_k,\dd\mu_0)e^{-\lambda_k t},
\end{align*}
and therefore
\begin{align*}
    f_i(t)
    :=F_i'(t)
    =\sum_{k\ge1} \lambda_k(u_k,p_i\rho)(u_k,\dd\mu_0)e^{-\lambda_k t}
\end{align*}
Applying Proposition~\ref{genint} and Theorem~\ref{estslow} yields
 \begin{align*}
    \P(\kappa_{n(N-j)}=i)
    \to\frac{(u_1,p_i\rho)}{(u_1,\rho)}
    =\frac{(u_1\rho,p_i)}{(u_1,\rho)}\quad\text{as }N\to\infty.
\end{align*}

Now, the solution to the forward Fokker-Planck equation is given by
\begin{align*}
    p(x,t)
    =\P(X(t)=\dd x\,|\,\tau>t)
    =\sum_{k\ge1}e^{-\lambda_k t}(u_k,\dd\mu_0)\rho(x)u_k(x).
\end{align*}
Hence, $u_1(x)\rho(x)/(u_1,\rho)$ is the quasi-stationary distribution (QSD), $q(x),$ defined by
\begin{align*}
    q(x)
    =\lim_{t\to\infty}\P(X(t)=\dd x\,|\,\tau>t)
    &=\lim_{t\to\infty}\frac{\P(X(t)=\dd x\cap\tau>t)}{\P(\tau>t)}\\
    &=\lim_{t\to\infty}\frac{\sum_{k\ge1}e^{-\lambda_k t}(u_k,\dd\mu_0)\rho(x)u_k(x)}{\sum_{k\ge1}(u_k,\rho)(u_k,\dd\mu_0)e^{-\lambda_k t}}\\
    &=\lim_{t\to\infty}\frac{e^{-\lambda_1 t}(u_1,\dd\mu_0)\rho(x)u_1(x)}{(u_1,\rho)(u_1,\dd\mu_0)e^{-\lambda_1 t}}\\
    &=\frac{\rho(x)u_1(x)}{(u_1,\rho)}.
\end{align*}

Summarizing, we have shown that
\begin{align}\label{slowint}
    \P(\kappa_{n(N-j)}=i)
    \to\int_U p_i(x)q(x)\,\dd x\quad\text{as }N\to\infty.
\end{align}

\paragraph{The case of drift-diffusion processes in one dimension.} For the one-dimensional example in which all the beliefs of all the agents evolve according to \eqref{simplesde}, we can compute the QSD, and find that 
\begin{align*}
    q(x)
    =\frac{\left(\pi ^2 D^2+\theta ^2 \mu ^2\right) \cos (\frac{\pi  x}{2 \theta }) e^{\frac{\mu  (\theta +x)}{2 D}}}{2 \pi  D^2 \theta  (e^{\frac{\theta  \mu }{D}}+1)}.
\end{align*}
Further, it is straightforward to show that the probability that a decider reaches $+\theta$ before $-\theta$ conditioned on the initial belief $x\in[-\theta,\theta]$ is
\begin{align*}
    p_1(x)
    :=\P(X(\tau)=+\theta)
    =\frac{1}{2} \Big(\coth \Big(\frac{\theta  \mu }{D}\Big)-1\Big) e^{\frac{\mu  (\theta -x)}{D}} \left(e^{\frac{\mu  (\theta +x)}{D}}-1\right)
\end{align*}
Therefore, applying \eqref{slowint} and explicitly computing the integral yields
 \begin{align*}
    \P(\kappa_{n(N-j)}=1)
    \to\int_{-\theta}^{\theta}p_1(x)q(x)\,\dd x
    =\frac{1}{1+e^{-\frac{\theta  \mu }{D}}}
    =p_1(0)\quad\text{as }N\to\infty.
\end{align*}
Hence, the slowest deciders out of $N\gg1$ deciders make a decision as if they were initially unbiased (i.e.\ as if $X(0)=0$).
\section{Proofs}\label{proofs}

\begin{proof}[Proof of Proposition~\ref{genint}]
Since $\{(\tau_{n},Z_{n})\}_{n\ge1}$ are identically distributed, we have that
\begin{align}\label{start0}
\begin{split}
\P(A_{n(j)})
&=\sum_{\textup{distinct indices}\atop n_{1},\dots,n_{N}\in\{1,\dots,N\}}\P(\max\{\tau_{n_{1}},\dots,\tau_{n_{j-1}}\}<\tau_{n_{j}}<\min\{\tau_{n_{j+1}},\dots,\tau_{n_{N}}\}\cap A_{n_{j}})\\
&=j{N\choose j}\P(\max\{\tau_{1},\dots,\tau_{j-1}\}<\tau_{j}<\min\{\tau_{j+1},\dots,\tau_{N}\}\cap A_{j}),
\end{split}
\end{align}
where the coefficient comes from noting that the number of terms in the sum is obtained by choosing the $j$ fastest FPTs out of $N$ and then choosing which of those $j$ will be the $j$th fastest. Define
\begin{align*}
\tau_{j}^{(A_{j})}
=\begin{cases}
\tau_{j} & \text{if $A_{j}$ occurs},\\
+\infty & \text{if $A_{j}$ does not occur},
\end{cases}
\end{align*}
so that if $j<N$,
\begin{align*}
&\P(\max\{\tau_{1},\dots,\tau_{j-1}\}<\tau_{j}<\min\{\tau_{j+1},\dots,\tau_{N}\}\cap A_{j})\\
&\quad=\P(\max\{\tau_{1},\dots,\tau_{j-1}\}<\tau_{j}^{(A_{j})}<\min\{\tau_{j+1},\dots,\tau_{N}\})
\end{align*}
To handle the case $j=N$, we can simply replace $+\infty$ by $-\infty$ in the definition of $\tau_{j}^{(A_{j})}$.

Since $\{\tau_{n}\}_{n\ge1}$ are iid, we have that
\begin{align*}
\P(\max\{\tau_{1},\dots,\tau_{j-1}\}<t)
=\P(\max\{\tau_{1},\dots,\tau_{j-1}\}\le t)
=[F(t)]^{j-1},
\end{align*}
where we have used that $F(t)$ is continuous. Similarly, 
\begin{align*}
\P(\min\{\tau_{j+1},\dots,\tau_{N}\}>t)
=[1-F(t)]^{N-j},
\end{align*}
Using that $\{\tau_{n}\}_{n\ge1}$ are independent, we have
\begin{align*}
G(t)
:&=\P(\max\{\tau_{1},\dots,\tau_{j-1}\}<t<\min\{\tau_{j+1},\dots,\tau_{N}\})\\
&=\P(\max\{\tau_{1},\dots,\tau_{j-1}\}<t)\P(t<\min\{\tau_{j+1},\dots,\tau_{N}\})\\
&=[F(t)]^{j-1}[1-F(t)]^{N-j}.
\end{align*}
Combining the above finally yields
\begin{align*}
\P(A_{n(j)})
&=j{N\choose j}\P(\max\{\tau_{1},\dots,\tau_{j-1}\}<\tau_{j}^{(A_{j})}<\min\{\tau_{j+1},\dots,\tau_{N}\})\\
&=j{N\choose j}\E[G(\tau_{j}^{(A_{j})})]\\
&=j{N\choose j}\int_{0}^{\infty}[F(t)]^{j-1}[1-F(t)]^{N-j}\,\dd F_{E}(t),
\end{align*}
which completes the proof.
\end{proof}

The proof of Theorem~\ref{detailprobj} is similar to the proof of Theorem~3 in \cite{linn2022extreme}.

\begin{proof}[Proof of Theorem~\ref{detailprobj}]
Define the integral from $t=a$ to $t=b$,
\begin{align*}
I_{a,b}
:=\int_{a}^{b}[F(t)]^{j-1}[1-F(t)]^{N-j}\,\dd F_{+}(t).
\end{align*}
Let $\eps\in(0,1)$. By the assumptions in Eq.~\eqref{s1}-\eqref{s2}, there exists a $\delta>0$ so that
\begin{align}
A_{-\eps}t^{p}e^{-C_{0}/t}
&\le F(t)
\le A_{+\eps}t^{p}e^{-C_{0}/t} \quad \text{for all }t\in(0,\delta),\label{lb199}\\
B_{-\eps}t^{q}e^{-C_{+}/t}
&\le F_{+}(t)
\le B_{+\eps}t^{q}e^{-C_{+}/t} \quad \text{for all }t\in(0,\delta),\label{lb299}
\end{align}
where $A_{\pm\eps}:=A(1\pm\eps)$ and $B_{\pm\eps}:=B(1\pm\eps)$. 
Using Eq.~\eqref{lb199} and integrating by parts yields
\begin{align}\label{ibp0}
\begin{split}
&I_{0,\delta}
\le \int_{0}^{\delta}(A_{+\eps}t^{p}e^{-C_{0}/t})^{j-1}(1-A_{-\eps}t^{p}e^{-C_{0}/t})^{N-j}\,\dd F_{+}(t)\\
&= (A_{+\eps}\delta^{p}e^{-C_{0}/\delta})^{j-1}(1-A_{-\eps}\delta^{p}e^{-C_{0}/\delta})^{N-j}F_{+}(\delta)\\
&
+ (N-j)\int_{0}^{\delta}(A_{+\eps}t^{p}e^{-C_{0}/t})^{j-1}(pt^{-1}+C_{0}t^{-2})A_{-\eps}t^{p}e^{-C_{0}/t}\big(1-A_{-\eps}t^{p}e^{-C_{0}/t}\big)^{N-j-1}F_{+}(t)\,\dd t\\
&
- (j-1)\int_{0}^{\delta}(A_{+\eps}t^{p}e^{-C_{0}/t})^{j-1}(pt^{-1}+C_{0}t^{-2})\big(1-A_{-\eps}t^{p}e^{-C_{0}/t}\big)^{N-j}F_{+}(t)\,\dd t\\
&\le (A_{+\eps}\delta^{p}e^{-C_{0}/\delta})^{j-1}(1-A_{-\eps}\delta^{p}e^{-C_{0}/\delta})^{N-j}F_{+}(\delta)\\
&
+ (N-j)\int_{0}^{\delta}(A_{+\eps}t^{p}e^{-C_{0}/t})^{j}(pt^{-1}+C_{0}t^{-2})\big(1-A_{-\eps}t^{p}e^{-C_{0}/t}\big)^{N-j-1}B_{+\eps}t^{q}e^{-C_{+}/t}\,\dd t\\
&
- (j-1)\int_{0}^{\delta}(A_{+\eps}t^{p}e^{-C_{0}/t})^{j-1}(pt^{-1}+C_{0}t^{-2})\big(1-A_{-\eps}t^{p}e^{-C_{0}/t}\big)^{N-j}B_{-\eps}t^{q}e^{-C_{+}/t}\,\dd t,
\end{split}
\end{align}
where we have used Eq.~\eqref{lb299} in the final inequality. 
The first term in the righthand side of Eq.~\eqref{ibp0} vanishes exponentially fast as $N\to\infty$. Using Proposition~\ref{p1} to find the large $N$ behavior of the second two terms in the righthand side of Eq.~\eqref{ibp0} and the fact that $I_{\delta,\infty}$ vanishes exponentially fast as $N\to\infty$ yields
\begin{align*}
\limsup_{N\to\infty}\frac{j{N\choose j}I_{0,\infty}}{{{\eta_j}}(\ln N)^{p{\beta}-q}N^{1-{\beta}}}
\le \frac{(1+\eps)}{(1-\eps)^{\beta}}.
\end{align*}
The analogous argument yields the lower bound
\begin{align*}
\liminf_{N\to\infty}\frac{j{N\choose j}I_{0,\infty}}{{{\eta_j}}(\ln N)^{p{\beta}-q}N^{1-{\beta}}}
\ge \frac{(1-\eps)}{(1+\eps)^{\beta}}.
\end{align*}
Since $\eps\in(0,1)$ is arbitrary, the proof is complete.
\end{proof}



\begin{proof}[Proof of Theorem~\ref{logprob}]
Define the integral from $t=a$ to $t=b$,
\begin{align*}
I_{a,b}
:=\int_{a}^{b}[F(t)]^{j-1}[1-F(t)]^{N-j}\,\dd F_{+}(t).
\end{align*}
 By Eq.~\eqref{lb}, there exists a $\delta>0$ so that
\begin{align}
e^{-(C_{0}+\eps)/t}\le
F(t)
&\le e^{-(C_{0}-\eps)/t}\quad\text{for all }t\in(0,\delta),\label{lb1j}\\
F_+(t)
&\le e^{-(C_+-\eps)/t}\quad\text{for all }t\in(0,\delta).\label{lb2j}
\end{align}
Using Eq.~\eqref{lb1j} and integrating by parts yields
\begin{align}\label{ibp055j}
\begin{split}
I_{0,\delta}
&\le \int_{0}^{\delta}e^{-(j-1)(C_{0}-\eps)/t}\big(1-e^{-(C_{0}+\eps)/t}\big)^{N-j}\,\dd F_+(t)\\
&= F_+(\delta)e^{-(j-1)(C_{0}-\eps)/\delta}\big(1-e^{-(C_{0}+\eps)/\delta}\big)^{N-j}\\
&\quad
+ (N-j)(C_{0}+\eps)\int_{0}^{\delta}e^{-(j-1)(C_{0}-\eps)/t}t^{-2}e^{-(C_{0}+\eps)/t}\big(1-e^{-(C_{0}+\eps)/t}\big)^{N-j-1}F_+(t)\,\dd t\\
&\quad
- \int_{0}^{\delta}(j-1)(C_{0}-\eps)t^{-2}e^{-(j-1)(C_{0}-\eps)/t}\big(1-e^{-(C_{0}+\eps)/t}\big)^{N-j}F_+(t)\,\dd t.
\end{split}
\end{align}
The first term in the righthand side of Eq.~\eqref{ibp055j} vanishes exponentially fast as $N\to\infty$. To handle the second term in the righthand side of Eq.~\eqref{ibp055j}, note that Eq.~\eqref{lb2j} implies that
\begin{align}\label{to9j}
\begin{split}
&\int_{0}^{\delta}e^{-(j-1)(C_{0}-\eps)/t}t^{-2}e^{-(C_{0}+\eps)/t}\big(1-e^{-(C_{0}+\eps)/t}\big)^{N-j-1}F_+(t)\,\dd t\\
&\quad\le\int_{0}^{\delta}e^{-(j-1)(C_{0}-\eps)/t}t^{-2}e^{-(C_{0}+\eps)/t}\big(1-e^{-(C_{0}+\eps)/t}\big)^{N-j-1}e^{-(C_+-\eps)/t}\,\dd t.
\end{split}
\end{align}
Since the third term in the righthand side of Eq.~\eqref{ibp055j} is nonpositive, applying Proposition~\ref{p1} to Eq.~\eqref{to9j} and using Eq.~\eqref{ibp055j} and the fact that $I_{\delta,\infty}$ vanishes exponentially fast as $N\to\infty$ completes the proof of Eq.~\eqref{ub}.

 If Eq.~\eqref{further} holds, then there exists a $\delta>0$ so that
\begin{align}\begin{split}\label{bounds98}
e^{-(C_{0}+\eps)/t}\le
F(t)
&\le e^{-(C_{0}-\eps)/t}\quad\text{for all }t\in(0,\delta),\\
e^{-(C_++\eps)/t}
\le F_+(t)
&\le e^{-(C_+-\eps)/t}\quad\text{for all }t\in(0,\delta).
\end{split}
\end{align}
Using Eq.~\eqref{bounds98} and integrating by parts yields
\begin{align}
\begin{split}\label{scareful}
    I_{0,\delta}
&\ge \int_{0}^{\delta}e^{-(j-1)(C_{0}+\eps)/t}\big(1-e^{-(C_{0}-\eps)/t}\big)^{N-j}\,\dd F_+(t)\\
&= F_+(\delta)e^{-(j-1)(C_{0}+\eps)/\delta}\big(1-e^{-(C_{0}-
\eps)/\delta}\big)^{N-j}\\
&\quad
+ (N-j)(C_{0}-\eps)\int_{0}^{\delta}e^{-(j-1)(C_{0}+\eps)/t}t^{-2}e^{-(C_{0}-\eps)/t}\big(1-e^{-(C_{0}-\eps)/t}\big)^{N-j-1}F_+(t)\,\dd t\\
&\quad
- \int_{0}^{\delta}(j-1)(C_{0}+\eps)t^{-2}e^{-(j-1)(C_{0}+\eps)/t}\big(1-e^{-(C_{0}-\eps)/t}\big)^{N-j}F_+(t)\,\dd t\\
&\ge F_+(\delta)e^{-(j-1)(C_{0}+\eps)/\delta}\big(1-e^{-(C_{0}-\eps)/\delta}\big)^{N-j}\\
&\quad
+ (N-j)(C_{0}-\eps)\int_{0}^{\delta}e^{-(j-1)(C_{0}+\eps)/t}t^{-2}e^{-(C_{0}-\eps)/t}\big(1-e^{-(C_{0}-\eps)/t}\big)^{N-j-1}e^{-(C_++\eps)/t}\,\dd t\\
&\quad
- \int_{0}^{\delta}(j-1)(C_{0}+\eps)t^{-2}e^{-(j-1)(C_{0}+\eps)/t}\big(1-e^{-(C_{0}-\eps)/t}\big)^{N-j}e^{-(C_+-\eps)/t}\,\dd t.
\end{split}
\end{align}
The first term in the righthand side of Eq.~\eqref{scareful} vanishes exponentially as $N\to\infty$. Using Proposition~\ref{p1} to estimate the second two terms in the righthand side of Eq.~\eqref{scareful} completes the proof.
\end{proof}

\begin{lemma}\label{intcalc}
For fixed $j\in\{0,1,\dots\}$, $c>0$, $\lambda>0$, and $\delta>0$, we have that
    \begin{align*}
        (N-j){N\choose N-j}\int_{1/\delta}^\infty\big[1-ce^{-\lambda t}]^{N-j-1}e^{-(j+1)\lambda t}\,\dd t
        \to\frac{1}{\lambda c^{j+1}}\quad\text{as }N\to\infty.
    \end{align*}
\end{lemma}

\begin{proof}[Proof of Lemma~\ref{intcalc}]
        Changing variables
    \begin{align*}
        u=1-ce^{-\lambda t},\quad \dd u=\lambda c e^{-\lambda t}\,\dd t
    \end{align*}
    yields
    \begin{align}
        \int_{1/\eps}^\infty\big[1-ce^{-\lambda t}]^{N-j-1}e^{-(j+1)\lambda t}\,\dd t
        &=\frac{1}{\lambda c^{j+1}}\int_{1-ce^{-\lambda t}}^1 u^{N-j-1}(1-u)^j\,\dd u\nonumber\\
        &=\frac{1}{\lambda c^{j+1}}\bigg[\frac{(N-j-1)! j!}{N!}-\int_0^{1-ce^{-\lambda t}} u^{N-j-1}(1-u)^j\,\dd u\bigg],\label{clear}
    \end{align}
    where we have used that $\int_0^1 u^{a-1}(1-u)^{b-1}\,\dd u=\Gamma(a)\Gamma(b)/\Gamma(a+b)$. Since the integral in Eq.~\eqref{clear} vanishes exponentially fast, the proof is complete.
\end{proof}


\begin{proof}[Proof of Theorem~\ref{estslow}]
Let $\eps\in(0,1)$. By assumption, there exists $\delta>0$ so that
\begin{align*}
    1-(1+\eps)ce^{-\lambda t}
    \le F(t)
    \le 1-(1-\eps)ce^{-\lambda t}\quad\text{for all }t\ge1/\delta,\\
    \lambda(1-\eps)c_ie^{-\lambda t}
    \le f_i(t)
    \le \lambda(1+\eps)c_ie^{-\lambda t}\quad\text{for all }t\ge1/\delta.
\end{align*}
Defining the integral from $t=a$ to $t=b$,
\begin{align*}
I_{a,b}
:=\int_{a}^{b}[F(t)]^{N-j-1}[1-F(t)]^j f_i(t)\,\dd t,
\end{align*}
we therefore have that
\begin{align*}
    &(1-\eps)^{j+1}\lambda c_ic^j\int_{1/\delta}^\infty\big[1-(1+\eps)ce^{-\lambda t}]^{N-j-1}e^{-(j+1)\lambda t}\,\dd t
    \le I_{1/\delta,\infty}\\
    &\quad\le(1+\eps)^{j+1}\lambda c_ic^j\int_{1/\delta}^\infty\big[1-(1-\eps)ce^{-\lambda t}]^{N-j-1}e^{-(j+1)\lambda t}\,\dd t.
\end{align*}
Since $I_{0,1/\delta}$ vanishes exponentially fast as $N\to\infty$, Lemma~\ref{intcalc} implies that
\begin{align*}
    \Big(\frac{1-\eps}{1+\eps}\Big)^{j+1}\frac{c_i}{c}
    &\le\liminf_{N\to\infty}(N-j){N\choose N-j}I_{0,\infty}\\
    &\le\limsup_{N\to\infty}(N-j){N\choose N-j}I_{0,\infty}
    \le\Big(\frac{1+\eps}{1-\eps}\Big)^{j+1}\frac{c_i}{c}.
\end{align*}
Since $\eps\in(0,1)$ is arbitrary, the proof is complete.
\end{proof}

\section{Numerical solutions} \label{numsol}

Numerical solutions were computed via trapezoidal quadrature on Eq.~\eqref{jj} in Proposition~\ref{genint}. In each set of dynamics, we rescaled the drift-diffusion process on $[-\theta,\theta]$ to the interval $[0,\ell]$. The probability density function for hitting the left boundary in this system is \cite{navarro2009}
\begin{align} \label{f0}
    f_0(t) := \frac{\textup{d}}{\textup{d}t} F_0(t) = \textup{exp}\Big(-\frac{\mu x_0}{2D} - \frac{\mu^2t}{4D}\Big) \frac{D}{\ell^2} \phi\Big(\frac{Dt}{\ell^2},\frac{x_0}{\ell}\Big),
\end{align}
where
\begin{align} \label{smallphi}
    \phi(s,w) := \begin{cases}
        \sum_{k=1}^{\infty} \textup{exp}(-k^2\pi^2s) 2k\pi \sin(k\pi w), \\
        (4\pi s^3)^{-1/2} \sum_{k=-\infty}^{\infty} (w+2k) \textup{exp}\Big( -\frac{(w+2k)^2}{4s} \Big).
    \end{cases}
\end{align}
The expressions in Eq.~\eqref{smallphi} are equivalent but have distinct utility: the top expansion converges quickly for large $s$ while the bottom expansion converges quickly for small $s$. Hence, we utilize both expressions to more accurately compute probabilities associated with slow and fast deciders, respectively.

Integrating Eq.~\eqref{f0} yields
\begin{align*}
    F_0(t) = \int_0^t f_0(t')\,\textup{d}t' = \textup{exp}\Big(-\frac{\mu x_0}{2D}\Big) \Phi\Big(\frac{Dt}{\ell^2},\frac{x_0}{\ell} \Big)
\end{align*}
with long- and short-time expansions of $\Phi(s,w)$ given by
\begin{align*}
    \Phi(s,w) &= \int_0^s \phi(s',w)\,\textup{d}s'\\
    &= \begin{cases}
        \sum_{k=1}^{\infty} \Big( 1 - \textup{exp}[-(b + k^2\pi^2)s] \Big) \frac{2k\pi}{b + k^2\pi^2} \sin(k\pi w), \\
        \sum_{k=-\infty}^{\infty} \frac{\textup{sgn}(2k+w)}{2} \Big( e^{-\sqrt{\frac{b}{D}}|2k+w|}\textup{erfc}\big( \frac{|2k+w|}{\sqrt{4s}} - \sqrt{bs} \big) + e^{\sqrt{\frac{b}{D}}|2k+w|}\textup{erfc}\big( \frac{|2k+w|}{\sqrt{4s}} + \sqrt{bs} \big)\Big)
    \end{cases}
\end{align*}
where $b = (\mu\ell/2D)^2$. By symmetry one can determine the corresponding probability density and cumulative distribution functions for hitting the right boundary. Altogether, we acquire long- and short-time expressions for the cumulative distribution function of an agent making a decision,
\begin{align*}
    F(t) := F_0(t) + F_1(t).
\end{align*}

Where numerical solutions are illustrated, we use the short-time expressions of $\phi$ and $\Phi$ for $10^{-10}\leq t\leq 1$ and the complementary long-time expressions for $1<t\leq 100$, discretizing each time interval into $10^3$ log-spaced points. We consider $10^3$ terms in each series expansion. Moreover, we take $\ell=1$ and unless otherwise stated $D=1$. Finally, where more than one but finitely many initial beliefs are considered, we scale the probability functions according to the corresponding initial distribution as outlined in Section~\ref{pe3}.

Specific details of figures with numerical solutions are as follows: In Fig \textcolor{blue}{1}B we illustrate in color Eq.~\eqref{j1} where $F_E = F$ as defined above with $X_{n(1)}(0) = y$. The black curve, which contains the remaining mass of the total probability, is computed as the sum of the colored curves subtracted from one. In Fig \textcolor{blue}{2}A-C we illustrate the probability that the first decider chooses the decision at $X(T_1)=\theta$ conditioned on having a particular initial bias. Hence, by definition of conditional probability, the numerical solutions are produced from quadrature on ratios of Eq.~\eqref{j1} with $F_E = F_1$ in the numerator and $F_E = F$ in the denominator with $X_{n(1)} = y$. The inset of Fig \textcolor{blue}{2}C is one minus the outset. In Fig \textcolor{blue}{2}D we illustrate Eq.~\eqref{jj} where $F_E=F$ and $S_{n(j)} = s$. In Fig \textcolor{blue}{3}B we illustrate the probability that the last decider chooses the decision at $X(T_N) = \theta$ conditioned on having a particular initial bias. Similar to Fig \textcolor{blue}{2}B, the numerical solutions are produced from quadrature on ratios of Eq.~\eqref{jN} with $F_E = F_1$ in the numerator and $F_E = F$ in the denominator with $X_{n(N)}(0) = y$.

\section{Agent-Based stochastic simulations}
{\em a. One-dimensional drift diffusion equation.} To test the analytical solutions, we solved Eq. (1) in the main text using the Euler-Maruyama method, which describes the evidence accumulation process preceding binary decisions. In this approximation scheme, the true solution to the stochastic differential equation is approximated by a Markov chain $Y$ constructed by setting $Y_0 = X(0)$ and updating $Y$ according to the iterative scheme
$$
Y_{n+1} = Y_n + \mu \Delta t + \sqrt{2D}\Delta W
$$
where $Y_n \equiv Y(n\Delta t)$ is the value of the Markov chain after the $n$th update, and the random variables $\Delta W$ are independent and identically distributed Gaussian random variables with mean 0 and variance $\Delta t$. The equations were integrated until the value of $Y_n$ exceeded $\pm \theta$. 

The temporal discretization, $\Delta t,$ is user-defined. As $N$ grows, the time to first decision decays slowly. Thus, for large $N$, $\Delta t$ must be taken to be sufficiently small for accurate representation of decision dynamics. For simulations here, we chose $\Delta t = 10^{-3}$ for $1 \leq N \leq 1000$. For $N > 1000$, we chose $\Delta t = N^{-1}$.

{\em b. Two-dimensional drift diffusion equation.} Decisions between three choices require a drift-diffusion model evolving on a planar domain~\cite{mcmillen2006dynamics}. Updating the discrete-time approximation of Eq.~(\ref{E:multalt}) for each observer (dropping the $i$ subscript) using Euler-Maruyama provides the following iterative scheme
\begin{align*}
    Y_{n+1}^1 &= Y_n^1 + \mu^1 \Delta t + \sqrt{2D} \Delta W^1, \\
    Y_{n+1}^2 &= Y_n^2 + \mu^2 \Delta t + \sqrt{2D} \Delta W^2,
\end{align*}
where $Y_n^j = Y^j(n \Delta t)$ is the value of the belief after the $n$th update, the random variables $\Delta W^j$ are Gaussian random variables with mean 0 and variance $\sigma^2$. Equations are integrated until the vector $\left( \begin{array}{c} Y_n^1 \\ Y_n^2 \end{array} \right)$ departs the triangular domain
\begin{align*}
    \{ (Y^1,Y^2) | Y^2<h \ \& \ Y^2>-2h(3Y^1+1) \ \& \ Y^2>2h(3Y^1-1) \}
\end{align*}
where $h = 1/2/\sqrt{3}$. Choices of each agent are determined by whether the agent crosses the $Y^2=h$ or $Y^2=-2h(3Y^1+1)$ or $Y^2=2h(3Y^1-1)$ boundary. For simulations again we use $\Delta t = 10^{-3}$ for $1 \leq N \leq 1000$ and $\Delta t = N^{-1}$ for $N > 1000$.

For the 2D case in an equilateral triangle, the threshold $\theta$ is taken to be equal to the length of the apothem---defined as a line from the center of a regular polygon at right angles to any of its sides. Hence, an unbiased agent begins at the centre of the equilateral triangle. We prescribe initial data for biased agents to be anywhere along an apothem except the centre of the triangle.

\bibliographystyle{plain}
\bibliography{refs}

\begin{thebibliography}{10}

\bibitem{Acemoglu2011}
Daron Acemoglu, Munther~A Dahleh, Ilan Lobel, and Asuman Ozdaglar.
\newblock Bayesian learning in social networks.
\newblock {\em The Review of Economic Studies}, 78(4):1201--1236, 2011.

\bibitem{bender2013}
CM~Bender and SA~Orszag.
\newblock {\em Advanced mathematical methods for scientists and engineers {I:
  A}symptotic methods and perturbation theory}.
\newblock Springer Science \& Business Media, 2013.

\bibitem{Bogacz2006}
Rafal Bogacz, Eric Brown, Jeff Moehlis, Philip Holmes, and Jonathan~D. Cohen.
\newblock {The physics of optimal decision making: A formal analysis of models
  of performance in two-alternative forced-choice tasks.}
\newblock {\em Psychological Review}, 113(4):700--765, 2006.

\bibitem{busemeyer1993decision}
Jerome~R Busemeyer and James~T Townsend.
\newblock Decision field theory: a dynamic-cognitive approach to decision
  making in an uncertain environment.
\newblock {\em Psychological review}, 100(3):432, 1993.

\bibitem{chittka2003bees}
Lars Chittka, Adrian~G Dyer, Fiola Bock, and Anna Dornhaus.
\newblock Bees trade off foraging speed for accuracy.
\newblock {\em Nature}, 424(6947):388--388, 2003.

\bibitem{cohen2011measuring}
Marlene~R Cohen and Adam Kohn.
\newblock Measuring and interpreting neuronal correlations.
\newblock {\em Nature neuroscience}, 14(7):811--819, 2011.

\bibitem{Gold02}
Joshua~I Gold and Michael~N Shadlen.
\newblock Banburismus and the brain: decoding the relationship between sensory
  stimuli, decisions, and reward.
\newblock {\em Neuron}, 36(2):299--308, 2002.

\bibitem{gold2007neural}
Joshua~I Gold and Michael~N Shadlen.
\newblock The neural basis of decision making.
\newblock {\em Annu. Rev. Neurosci.}, 30:535--574, 2007.

\bibitem{grebenkov2020single}
Denis~S Grebenkov, Ralf Metzler, and Gleb Oshanin.
\newblock From single-particle stochastic kinetics to macroscopic reaction
  rates: fastest first-passage time of n random walkers.
\newblock {\em New Journal of Physics}, 22(10):103004, 2020.

\bibitem{karamched2020heterogeneity}
Bhargav Karamched, Megan Stickler, William Ott, Benjamin Lindner, Zachary~P
  Kilpatrick, and Kre{\v{s}}imir Josi{\'c}.
\newblock Heterogeneity improves speed and accuracy in social networks.
\newblock {\em Physical review letters}, 125(21):218302, 2020.

\bibitem{Karamched20}
Bhargav Karamched, Simon Stolarczyk, Zachary~P Kilpatrick, and Kre\v{s}imir
  Josi\'{c}.
\newblock Bayesian evidence accumulation on social networks.
\newblock {\em SIAM Journal on Applied Dynamical Systems}, 19(3):1884--1919,
  2020.

\bibitem{kiani2008bounded}
Roozbeh Kiani, Timothy~D Hanks, and Michael~N Shadlen.
\newblock Bounded integration in parietal cortex underlies decisions even when
  viewing duration is dictated by the environment.
\newblock {\em Journal of Neuroscience}, 28(12):3017--3029, 2008.

\bibitem{koriat2012two}
Asher Koriat.
\newblock When are two heads better than one and why?
\newblock {\em Science}, 336(6079):360--362, 2012.

\bibitem{latimer2015single}
Kenneth~W Latimer, Jacob~L Yates, Miriam~LR Meister, Alexander~C Huk, and
  Jonathan~W Pillow.
\newblock Single-trial spike trains in parietal cortex reveal discrete steps
  during decision-making.
\newblock {\em Science}, 349(6244):184--187, 2015.

\bibitem{lawley2020dist}
S~D Lawley.
\newblock Distribution of extreme first passage times of diffusion.
\newblock {\em Journal of Mathematical Biology}, 2020.

\bibitem{Lawley20b}
S.~D. Lawley.
\newblock Universal formula for extreme first passage statistics of diffusion.
\newblock {\em Physical Review E}, (1):012413, 2020.

\bibitem{linn2022extreme}
Samantha Linn and Sean~D Lawley.
\newblock Extreme hitting probabilities for diffusion.
\newblock {\em Journal of Physics A: Mathematical and Theoretical},
  55(34):345002, 2022.

\bibitem{madrid2020comp}
Jacob~B Madrid and Sean~D Lawley.
\newblock Competition between slow and fast regimes for extreme first passage
  times of diffusion.
\newblock {\em Journal of Physics A: Mathematical and Theoretical},
  53(33):335002, 2020.

\bibitem{mann2018collective}
Richard~P Mann.
\newblock Collective decision making by rational individuals.
\newblock {\em Proceedings of the National Academy of Sciences},
  115(44):E10387--E10396, 2018.

\bibitem{marshall2017individual}
James~AR Marshall, Gavin Brown, and Andrew~N Radford.
\newblock Individual confidence-weighting and group decision-making.
\newblock {\em Trends in ecology \& evolution}, 32(9):636--645, 2017.

\bibitem{mcmillen2006dynamics}
Tyler McMillen and Philip Holmes.
\newblock The dynamics of choice among multiple alternatives.
\newblock {\em Journal of Mathematical Psychology}, 50(1):30--57, 2006.

\bibitem{mossel2014opinion}
Elchanan Mossel and Omer Tamuz.
\newblock Opinion exchange dynamics.
\newblock {\em arXiv preprint arXiv:1401.4770}, 2014.

\bibitem{mulder2012}
Martijn~J Mulder, Eric-Jan Wagenmakers, Roger Ratcliff, Wouter Boekel, and
  Birte~U Forstmann.
\newblock Bias in the brain: a diffusion model analysis of prior probability
  and potential payoff.
\newblock {\em Journal of Neuroscience}, 32(7):2335--2343, 2012.

\bibitem{navarro2009}
Daniel~J Navarro and Ian~G Fuss.
\newblock Fast and accurate calculations for first-passage times in wiener
  diffusion models.
\newblock {\em Journal of Mathematical Psychology}, 53(4):222--230, 2009.

\bibitem{newsome1989neuronal}
William~T Newsome, Kenneth~H Britten, and J~Anthony Movshon.
\newblock Neuronal correlates of a perceptual decision.
\newblock {\em Nature}, 341(6237):52--54, 1989.

\bibitem{Note1}
For an ideal observer if only measurement noise is present then drift and
  diffusion are equal in the continuum limit, $\mu _i=D_i=m_i$~\cite
  {Bogacz2006,veliz16}. We consider the more general case where internal noise
  can increase $D_i>\mu _i$.

\bibitem{Note2}
This is the distribution of beliefs conditioned on the absence of a decision
  after a long time.

\bibitem{Note3}
We interpret beliefs as log-likelihood ratios. Therefore, with $k$ alternatives
  beliefs evolve in $k-1$ dimensions.

\bibitem{pavliotis2014}
G~A Pavliotis.
\newblock {\em Stochastic processes and applications: diffusion processes, the
  Fokker-Planck and Langevin equations}, volume~60.
\newblock Springer, 2014.

\bibitem{Ratcliff1978theory}
Roger Ratcliff.
\newblock A theory of memory retrieval.
\newblock {\em Psychological review}, 85(2):59, 1978.

\bibitem{reina2023asynchrony}
Andreagiovanni Reina, Thomas Bose, Vaibhav Srivastava, and James~AR Marshall.
\newblock Asynchrony rescues statistically optimal group decisions from
  information cascades through emergent leaders.
\newblock {\em Royal Society Open Science}, 10(3):230175, 2023.

\bibitem{shadlen1996}
Michael~N Shadlen and William~T Newsome.
\newblock Motion perception: seeing and deciding.
\newblock {\em Proceedings of the national academy of sciences},
  93(2):628--633, 1996.

\bibitem{stickler2023impact}
Megan Stickler, William Ott, Zachary~P Kilpatrick, Kre{\v{s}}imir Josi{\'c},
  and Bhargav~R Karamched.
\newblock Impact of correlated information on pioneering decisions.
\newblock {\em Physical Review Research}, 5(3):033020, 2023.

\bibitem{swets1961decision}
John~A Swets, Wilson~P Tanner~Jr, and Theodore~G Birdsall.
\newblock Decision processes in perception.
\newblock {\em Psychological review}, 68(5):301, 1961.

\bibitem{tajima2019optimal}
Satohiro Tajima, Jan Drugowitsch, Nisheet Patel, and Alexandre Pouget.
\newblock Optimal policy for multi-alternative decisions.
\newblock {\em Nature neuroscience}, 22(9):1503--1511, 2019.

\bibitem{tump2022avoiding}
Alan~N Tump, Max Wolf, Pawel Romanczuk, and Ralf~HJM Kurvers.
\newblock Avoiding costly mistakes in groups: the evolution of error management
  in collective decision making.
\newblock {\em PLoS Computational Biology}, 18(8):e1010442, 2022.

\bibitem{uchida2003speed}
Naoshige Uchida and Zachary~F Mainen.
\newblock Speed and accuracy of olfactory discrimination in the rat.
\newblock {\em Nature neuroscience}, 6(11):1224--1229, 2003.

\bibitem{veliz16}
A.~Veliz-Cuba, Z.~P. Kilpatrick, and K.~Josi{\'c}.
\newblock Stochastic models of evidence accumulation in changing environments.
\newblock {\em SIAM Review}, 58(2):264--289, 2016.

\bibitem{wang2012neural}
Xiao-Jing Wang.
\newblock Neural dynamics and circuit mechanisms of decision-making.
\newblock {\em Current opinion in neurobiology}, 22(6):1039--1046, 2012.

\bibitem{zandbelt2014response}
Bram Zandbelt, Braden~A Purcell, Thomas~J Palmeri, Gordon~D Logan, and
  Jeffrey~D Schall.
\newblock Response times from ensembles of accumulators.
\newblock {\em Proceedings of the National Academy of Sciences},
  111(7):2848--2853, 2014.

\end{thebibliography}
\end{document}